\documentclass[11pt,a4paper,reqno,oneside]{amsart}
\usepackage[utf8]{inputenc}
\usepackage[T1]{fontenc}
\usepackage[a4paper]{geometry}
\usepackage{lmodern,textcomp,amsmath,amssymb,mathtools,braket,authblk,amsthm,csquotes,bm}
\usepackage{amssymb,comment,stmaryrd}
\usepackage[english]{babel}
\usepackage[foot]{amsaddr}
\usepackage[draft]{fixme}

\newtheorem{theorem}{Theorem}[section]

\newtheorem{lemma}[theorem]{Lemma}
\newtheorem{proposition}[theorem]{Proposition}
\theoremstyle{definition}
\newtheorem{assumption}[theorem]{Assumption}
\newtheorem{definition}[theorem]{Definition}

\theoremstyle{remark}
\newtheorem{remark}[theorem]{Remark}
\newtheorem{example}[theorem]{Example}
\AtBeginEnvironment{example}{%
	\pushQED{\qed}%
}
\AtEndEnvironment{example}{\popQED\endexample}
\AtBeginEnvironment{remark}{%
	\pushQED{\qed}%
}
\AtEndEnvironment{remark}{\popQED\endexample}

\newcommand{\di}{\mathrm{d}}
\newcommand{\hilb}{\mathcal{H}}
\newcommand{\fock}{\mathcal{F}}
\newcommand{\hfrak}{\mathfrak{H}}
\newcommand{\e}{\mathrm{e}}

\renewcommand{\Im}{\operatorname{Im}}
\renewcommand{\Re}{\operatorname{Re}}
\renewcommand{\a}[1]{a\!\left(#1\right)}
\newcommand{\adag}[1]{a^*\!\left(#1\right)}
\newcommand{\ii}{\mathrm{i}}
\newcommand{\dOmega}{\mathrm{d}\Gamma(\omega)}

\newcommand{\fh}{\mathfrak{h}}
\newcommand{\I}{\textnormal{I}}
\newcommand{\II}{\textnormal{II}}
\newcommand{\III}{\textnormal{III}}
\newcommand{\bare}{\textnormal{bare}}
\newcommand{\dom}{\operatorname{Dom}}
\newcommand{\ran}{\operatorname{Ran}}

\makeatletter
\def\smalloverbrace#1{\mathop{\vbox{\m@th\ialign{##\crcr\noalign{\kern3\p@}%
				\tiny\downbracefill\crcr\noalign{\kern3\p@\nointerlineskip}%
				$\hfil\displaystyle{#1}\hfil$\crcr}}}\limits}
\makeatother

\usepackage[unicode=true,
bookmarks=true,bookmarksopen=false,
breaklinks=false,pdfborder={0 0 0},colorlinks=true]
{hyperref}

\usepackage{xcolor}
\definecolor{cblue}{rgb}{0.16, 0.32, 0.75}
\definecolor{cred}{rgb}{0.7, 0.11, 0.11}
\hypersetup{%
	,linkcolor=cred
	,citecolor=cblue
	,urlcolor=black
}

\usepackage[
backend=bibtex,
style=numeric-comp,
giveninits=true,
maxbibnames=299,
natbib=true,
doi=false,
isbn=false,
url=false,
sorting=nyt
]
{biblatex}
\renewbibmacro{in:}{}
\DeclareFieldFormat[misc]{title}{``#1''\isdot}
\DeclareFieldFormat{journaltitle}{#1\isdot}
\DeclareFieldFormat[article,periodical]{volume}{\mkbibbold{#1}}
\DeclareFieldFormat{pages}{#1}
\AtEveryBibitem{%
	\clearfield{number}}

\usepackage{orcidlink}

\newcommand\blfootnote[1]{%
  \begingroup
  \renewcommand\thefootnote{}\footnote{#1}%
  \addtocounter{footnote}{-1}%
  \endgroup
}

\usepackage{todonotes}

\title[Renormalization of GSB models with critical UV divergences]{Renormalization of generalized spin--boson models with critical ultraviolet divergences}

\author{Benjamin Alvarez\textsuperscript{\,1}\hspace{3pt}\hspace{1pt}}
\author{Sascha Lill\textsuperscript{\,2}\hspace{3pt}\orcidlink{0000-0002-9474-9914}\hspace{1pt}}
\author{Davide Lonigro\textsuperscript{\,3}\hspace{3pt}\orcidlink{0000-0002-0792-8122}\hspace{1pt}}
\author{Javier Valentín Martín\textsuperscript{\,4}}

\address{\footnotesize \textsuperscript{1}Aix Marseille Univ, Univ Toulon, CNRS, CPT, Marseille, France}
\address{\footnotesize \textsuperscript{2}Department of Mathematical Sciences, Universitetsparken 5, DK-2100 Copenhagen, Denmark}
\address{\footnotesize \textsuperscript{3}Department Physik, Friedrich-Alexander-Universität Erlangen-Nürnberg, Staudtstraße 7, 91058 Erlangen, Germany}
\address{\footnotesize \textsuperscript{4}Universität Paderborn, Institut für Mathematik, Institut für Photonische Quantensysteme, Warburger Str. 100, 33098 Paderborn, Germany}
\email{\footnotesize \href{mailto:benjamin.alvarez@univ-tln.fr}{\texttt{benjamin.alvarez@univ-tln.fr}}}
\email{\footnotesize \href{mailto:sali@math.ku.dk}{\texttt{sali@math.ku.dk}}}
\email{\footnotesize \href{mailto:davide.lonigro@fau.de}{\texttt{davide.lonigro@fau.de}}}
\email{\footnotesize \href{mailto:javiervm@math.uni-paderborn.de}{\texttt{javiervm@math.uni-paderborn.de}}}

\begin{document}
	
\maketitle\thispagestyle{empty}

\vspace{-0.75cm}
    
\begin{abstract}
    We provide a rigorous construction of generalized spin--boson models with commuting transition matrices and form factors exhibiting critical ultraviolet (UV) divergences. That is, we cover all divergences where a self-energy renormalization, but no non-Fock representation, is required. Our method is based on a direct definition of the renormalized Hamiltonian on a sufficiently large test domain, followed by a Friedrichs extension.
    We then prove that this Hamiltonian coincides with the one obtained by cut-off renormalization.
    Furthermore, we show that for specific supercritical cases, i.e., when a non-Fock representation is required, the renormalized Hamiltonian is trivial.
\end{abstract}



\section{Introduction}\label{sec:introduction}

Generalized spin--boson (GSB) models describe a quantum mechanical system, usually finite-dimensional, linearly coupled with a quantum bosonic field~\cite{arai1997existence}. The Hilbert space of the total system is $\hfrak=\mathbb{C}^D\otimes\fock$, where
\begin{equation}\label{eq:fock}
    \fock:=\mathcal{F}^+(\hilb)=\bigoplus_{n=0}^\infty S_n\hilb^{\otimes n}
\end{equation}
is the Bose--Fock space over a single-particle space $\hilb$, with $ S_n $ being the symmetrization operator. In applications, $\hilb$ is usually taken as a space of square-integrable functions over some $\sigma$-finite measure space $(X,\Sigma,\mu)$; to fix ideas we will henceforth adopt this choice, simply denoting it by $\hilb=L^2(X)$ in what follows. The Hamiltonian for GSB models formally reads
\begin{equation}\label{eq:gsb}
    H^{\bare}
    = K \otimes 1
        + 1 \otimes \di \Gamma (\omega)
        +\sum_{j=1}^N \left(B_j\otimes \adag{f_j}+B^*_j\otimes\a{f_j}\right) \;,
\end{equation}
with $K=K^* \in \mathbb{C}^{D \times D}$ describing the free dynamics of the spin system, $ \omega: X \to \mathbb{R} $ being the boson dispersion relation, and where the transition matrices $B_1,\dots,B_N\in\mathbb{C}^{D \times D}$, and the form factors $f_1,\dots,f_N\in L^2(X)$ model the $ N \in \mathbb{N} $ system--field couplings. Here, $\a{f},\adag{f}$ denote the annihilation and creation operator associated with a given $f\in L^2(X)$. We refer to Section~\ref{sec:setting} for all due mathematical details. The spectral properties of such models, as well as their properties in the presence of infrared divergences, have attracted a lot of attention in the past decades, cf.~\cite{hirokawa2001remarks,hubner1995spectral,hirokawa1999expression,arai1990asymptotic,amann1991ground,davies1981symmetry,fannes1988equilibrium,hubner1995radiative,reker2020existence,hasler2021existence,lonigro2022generalized} for the case $N=1$ and~\cite{arai2000essential,arai2000ground,falconi2015self,takaesu2010generalized,teranishi2015self,teranishi2018absence,lonigro2023self} for $N>1$.\bigskip

In this paper, we will focus on the \textit{ultraviolet problem} (UV) for GSB models, for which we assume $ \omega \ge m > 0 $. This entails considering the case in which the system--field coupling is formally described by form factors whose norm diverges:
\begin{equation}\label{eq:normalization}
    \int_X|f_j(k)|^2\;\mathrm{d} k =\infty
\end{equation}
for some $j=1,\dots,N$. Note that here and in the rest of the article we write $ \mathrm{d} \mu(k) =: \mathrm{d} k $ for simplicity. Such divergences are frequently encountered in applications: UV-divergent form factors are used to model short-range interactions or as toy models of perfectly Markovian open quantum systems; we refer to Section~\ref{sec:physics} for a brief related discussion. 
When Eq.~\eqref{eq:normalization} holds, a direct interpretation of $ H^{\bare} $ as a legitimate self-adjoint operator on $\hfrak$ is not immediate, as the corresponding creation operator $\adag{f_j}$ is not densely defined on $\fock$. As such, a careful mathematical study---often involving a \textit{renormalization} procedure---is required in order to construct a self-adjoint \emph{renormalized operator} $ H $ out of $ H^{\bare} $.

Reprising the discussion in~\cite{lill2025self}, to which we refer for additional details, it is convenient to classify UV divergences in case $ \omega(k) \ge m > 0 $ as follows:
\begin{itemize}
    \item \textbf{Case 1}: $\int|f_j|^2=\infty$ but $\int\omega^{-1}|f_j|^2<\infty$ for all $j$.
     \item \textbf{Case 2}: $\int\omega^{-1}|f_j|^2=\infty$ but $\int\omega^{-2}|f_j|^2<\infty$ for all $j$.
     \item \textbf{Case 3}: $\int\omega^{-2}|f_j|^2=\infty$.
\end{itemize}
Additionally, the boundary point of Case 2, where $ \int \omega^{-2} |f_j|^2 < \infty $, while $ \int \omega^{-2+\varepsilon} |f_j|^2 = \infty $ for any $ \varepsilon > 0 $, is also referred to as the \emph{UV critical case} in~\cite{hinrichs2025ultraviolet,lampart2023hamiltonians}. Correspondingly, all weaker divergences of Case 1 or 2 are called \emph{sub-critical}, and all divergences of Case 3 are called \emph{super-critical}.

\subsection{Previous results}

Rigorous results on either case for GSB models are relatively recent, either for specific choices of the operators $B_1,\dots,B_N$ or more general. In Case 1, for any choice of the parameters, Eq.~\eqref{eq:gsb} can be directly interpreted in the sense of quadratic forms, cf.~\cite[Proposition 4.2]{lonigro2022generalized}, so that no renormalization procedure is necessary; furthermore, under a UV cut-off, norm resolvent convergence to said operator is achieved. An explicit characterization of the domain of this operator was obtained by two of the authors in~\cite{lill2025self} for normal and commuting $B_1,\dots,B_N$, as an application of abstract results on singular perturbations by Posilicano~\cite{posilicano2020self}. 

Additionally, one of us obtained renormalization results for the so-called rotating-wave spin--boson model, cf.~Section~\ref{sec:physics}, in Cases 1 and 2 for the single-atom case~\cite{lonigro2022generalized}, and in Case 1 for its multi-atom generalization~\cite{lonigro2022generalized,lonigro2022renormalization,lonigro2023self}: this was obtained by means of an explicit computation of the resolvent. UV cut-offs are there shown to produce either strong or norm resolvent convergence, depending on the badness of the divergence and the strength of the coupling. More recently, Hinrichs, Lampart and one of the authors~\cite{hinrichs2025ultraviolet} extended these results by successfully renormalizing GSB models with $N=1$, with a single operator $B$ which is either normal or nilpotent (or a commuting combinations of those), via a proof technique based on the interior--boundary conditions (IBC) method, see below. Finally, in Case 3 it was shown by Dam and Møller~\cite{dam2020asymptotics} that certain spin--boson models with $D=2$ exhibit \textit{triviality}, that is, they renormalize to a direct sum of two van Hove Hamiltonians, so that no system--field coupling is retained in the UV limit. Furthermore, renormalization in this case has been very recently achieved for $ N = 1 $ and $ B^* B = B B^* $~\cite{falconi2025nontrivial}, see also~\cite{falconi2025ultraviolet}.

An extensive mathematical literature addresses the UV renormalization problem in other models of matter–field interaction. Successful renormalization schemes in Cases 1 and 2 have been developed in seminal works such as~\cite{nelson1964interaction, eckmann1970model, frohlich1973infrared, sloan1974polaron}, and in more recent contributions~\cite{alvarez2021ultraviolet,alvarez2023ultraviolet,griesemer2016self, griesemer2018domain}. In Case 3, a breakthrough result was obtained by Gross~\cite{gross1973relativistic} for a closely related model. We refer to~\cite[Sect.~1.3]{lill2022time} for a broader survey of this literature.
Additionally, a novel renormalization method based on so-called interior–boundary conditions (IBC) has been introduced in recent years~\cite{teufel2021hamiltonians, teufel2016avoiding, lampart2018particle, lampart2019nelson, lampart2019nonrelativistic, schmidt2021massless, lampart2020renormalized,lampart2023hamiltonians,hinrichs2025ultraviolet}. For a detailed review of the IBC framework and recent developments, we refer to~\cite[Sect.~1.4]{lill2022time}.
Let us also point out that there exists a very extensive mathematical literature on the construction of relativistic quantum field theory (QFT) models, where one encounters similar UV problems as in Case 3, but where the Lorentz invariance offers further renormalization techniques, see~\cite{dedushenko2023snowmass} for an exhaustive literature review.

\subsection{Our results (informal)}
The main results of this paper (Theorems~\ref{thm:HE2def_T}--\ref{thm:normresolvent}) concern the renormalization of GSB models exhibiting Case 2 divergences. We work under Assumption~\ref{as:strong} on the model, which includes the following requirement on the matrices $B_j$:
\begin{equation}\label{eq:bj_commute}
    [B_j, B_{j'}] = 0 \;.
\end{equation}
This assumption is indeed satisfied in the physically relevant case of each $ B_j $ corresponding to a different sub-system (e.g., qubit, atom, or molecule). No other constraint on either the number $N$ of operators $B_1,\dots,B_N\in\mathbb{C}^{D \times D}$ or their structure, nor the dimension $D$ shall be imposed. Under this assumption, we prove the following facts:
\begin{itemize}
    \item There exists a self-adjoint operator $H$, semi-bounded from below, whose action on its domain agrees with the formal expression~\eqref{eq:gsb} modulo the subtraction of a properly chosen counterterm (Theorem~\ref{thm:HE2def_T});
    \item The Hamiltonian $H$ can be obtained as the norm resolvent limit of a family of cut-off Hamiltonians, obtained by replacing each UV-divergent form factor $f_j$ with a sequence $(f_{j,\Lambda})_{\Lambda\in\mathbb{N}} \subset L^2(X)$ (Theorem~\ref{thm:normresolvent}).
\end{itemize}
Our proof will be based on an energy counterterm applied at the level of the quadratic form uniquely associated with $H$: by doing so, one obtains a quadratic form densely defined on a domain different from the original form domain of $H$, which can nonetheless be completely characterized by means of generalized Weyl operators applied to a displaced vacuum vector. By providing a successful renormalization scheme in Case 2 for arbitrary values of $N$, our results extend both those in~\cite{lill2025self}, where Case 1 is addressed for generic $N$ and additional constraints on the spin operators are imposed on top of~\eqref{eq:bj_commute}, and those in~\cite{hinrichs2025ultraviolet}, where Case 2 is addressed for $N=1$ with a normal or nilpotent spin operator.

As a secondary result (Proposition~\ref{prop:triviality}), we extend the results by Dam and Møller to a larger class of GSB models exhibiting Case 3 divergences, by finding sufficient conditions for GSB models to exhibit triviality. The question whether it is possible to obtain a nontrivial Case 3 renormalization of certain GSB models remains open, and shall be the subject of future research.

\subsection{GSB models in physics}\label{sec:physics}

GSB models naturally arise in many branches of quantum mechanics~\cite{weiss2012quantum,leggett1987dynamics,breuer2002theory,ingold2002path,grifoni1999dissipation,thorwart2004dynamics,clos2012quantification,costi2003entanglement,larson2021jaynes}, ranging from quantum optics---where they serve as a simplified but useful description of interaction between light and matter---to quantum information and open system theory, where they are used as toy models to investigate fundamental properties of the influence of the environment on finite-dimensional systems, like non-Markovianity, decoherence, and emergence of non-classical behavior~\cite{bloch2012quantum,blatt2012quantum,porras2008mesoscopic,vogel1995nonlinear,phoenix1991establishment,hao2013dynamics,ge2010quantum,lemmer2018trapped,leppakangas2018quantum,wipf1987influence,suarez1991hydrogen,guinea1985bosonization}. While Case 1 divergences are typically encountered in the description of low-dimensional optical models, e.g., waveguide quantum electrodynamics~\cite{facchi2016bound,facchi2019bound,lonigro2021stationary}, Case 2 divergences appear when describing spin--boson models in the so-called Markovian limit, where the spin undergoes pure exponential decay and memory effects are totally suppressed~\cite{lonigro2022quantum,lonigro2022regression,chruscinski2023markovianity}.

In the case $D=2$ (spin--boson models), one generally chooses a representation where the energy of the two-level system, or spin, is described by a diagonal matrix,
\begin{equation}
    K=\frac{\omega_0}{2}\sigma_z,\qquad\omega_0\geq0 \;,
\end{equation}
with $\sigma_z$ being the third Pauli matrix; as such, $\omega_0$ represents the energy splitting between the ``excited'' and ``ground'' energy levels of the spin. The qubit--field interaction is then most generally described by a Hermitian matrix $B=B^*$, which can be expanded in the basis of $\mathbb{C}^{2\times 2}$ obtained by Pauli matrices $\{I,\sigma_x,\sigma_y,\sigma_z\}$:
\begin{equation}
    B=\alpha_0I+\vec{\alpha}\cdot\vec{\sigma}=\alpha_0I+\alpha_x\sigma_x+\alpha_y\sigma_y+\alpha_z\sigma_z \;,
\end{equation}
where $\alpha_0,\alpha_x,\alpha_y,\alpha_z$ are real numbers which determine the nature of the spin--boson coupling:
\begin{itemize}
    \item $\alpha_z$ determines the longitudinal component of the coupling, which is responsible for the \textit{decoherence} of the qubit and does not induce transitions between the two energy levels;
    \item $\alpha_x$ and $\alpha_y$ determine the transversal component of the coupling, which is instead responsible for \textit{transitions} between the two qubit energy levels. Usually one chooses a representation in which either $\alpha_x$ or $\alpha_y$, typically the latter, is zero. Spin--boson models with $B\propto\sigma_x$ reduce to the well-known \textit{quantum Rabi model}~\cite{jaynes1963comparison} in the case of a monochromatic boson field, i.e. $ L^2(X)\simeq\mathbb{C}$.
    \item $\alpha_0$ controls the equilibrium displacement of the boson modes.
\end{itemize}
While in the mathematical literature the name ``spin--boson model'' is usually reserved for the case $B\propto\sigma_x$, i.e. with purely transversal coupling, in physics, the name is alternatively used for the general case discussed above, or for any of the particular cases one obtains by only retaining one of the parameters $\alpha_0,\alpha_x,\alpha_y,\alpha_z$.

A further variable to be taken into account, which also introduces the necessity of studying the case of non-Hermitian \textit{and} non-normal $B$, is the rotating-wave approximation (RWA), which consists in the following substitution:
\begin{eqnarray}
     \sigma_x\otimes\left(\a{f}+\adag{f}\right)&=&(\sigma_++\sigma_-)\otimes\a{f}+(\sigma_++\sigma_-)\otimes\adag{f}\nonumber\\
     &&\overset{\text{RWA}}{\longrightarrow}\sigma_+\otimes\a{f}+\sigma_-\otimes\adag{f},
\end{eqnarray}
where $\sigma_\pm=\sigma_x\pm\ii\sigma_y$ are the ladder operators on the spin degree of freedom. The resulting model---sometimes referred to as the rotating-wave spin--boson model---is exactly solvable, in the sense that a closed expression for its resolvent can be obtained. For a monochromatic field, it reduces to the well-known Jaynes--Cummings model~\cite{jaynes1963comparison,shore1993jaynes,larson2021jaynes}. In the monochromatic case, this approximation is heuristically justified by the fact that, in the interaction picture with respect to the free operator, the terms $\sigma_-\otimes\a{f}$ and $\sigma_+\otimes\adag{f}$ are multiplied by terms that become quickly oscillating, and thus negligible on average, as $\omega_0\to\infty$; a rigorous justification of this approximation was only obtained relatively recently in~\cite{burgarth2024taming} for a single spin coupled to a monochromatic field, and in~\cite{richter2024quantifying} for multiple spins.

We finally point out that in the physics literature one usually considers boson fields with discrete modes $(\omega_k)_k$ indexed by some quasimomentum $k$, and defines the so-called \textit{spectral density} $J(E)$, formally defined by $J(E)=\sum_k|f_k|^2\delta(E-\omega_k)$, with $\omega_k$ being the dispersion of the field at the $k$-th mode; a continuum limit is then considered, so that $J(E)$ becomes a positive-valued continuous function defined on $[m,\infty)$. One usually proceeds studying the properties of the model for certain customary choices of $J(E)$. In our mathematical framework, this can be reproduced by choosing the measure space $(X,\Sigma,\mu)$ as the Lebesgue space on $\mathbb{R}$, and the form factor chosen in such a way that $|f(E)|^2=J(E)$. As such, the strength of the UV divergence in the model is entirely determined by the behavior of $J(E)$ as $E\to\infty$: the model is UV-divergent if $\int J
(E)\;\mathrm{d}E=+\infty$, and we distinguish
\begin{itemize}
    \item \textbf{Case 1}: $\int J(E)\;\mathrm{d}E=\infty$ but $\int E^{-1}J(E)\;\mathrm{d}E<\infty$.
    \item \textbf{Case 2}: $\int E^{-1}J(E)\;\mathrm{d}E=\infty$ but $\int E^{-2}J(E)\;\mathrm{d}E<\infty$.
    \item \textbf{Case 3}: $\int E^{-2}J(E)\;\mathrm{d}E=\infty$.
\end{itemize}
Similar arguments can be made in the case of multiple form factors.

\subsection{Outline of the paper} The present paper is organized as follows. In Section~\ref{sec:setting} we introduce the notation used throughout the paper, and state our results: Theorem~\ref{thm:HE2def_T}, concerning the existence of a renormalized GSB Hamiltonian in Case 2; Theorem~\ref{thm:normresolvent}, where we identify this Hamiltonian as the norm resolvent limit of a family of cut-off Hamiltonians; and Proposition~\ref{prop:triviality} about the triviality of certain GSB Hamiltonians in Case 3. In Section~\ref{sec:proof_renormalization_weylfree} we prove Theorem~\ref{thm:HE2def_T}, and in Section~\ref{sec:proof_normresolvent} we prove Theorem~\ref{thm:normresolvent}. Finally, Proposition~\ref{prop:triviality} is proven in Appendix~\ref{sec:triviality}.

\subsection*{Acknowledgments} SL and DL thank Benjamin Hinrichs for our scientific discussions and exchanges of ideas. BA thanks Jacob M\o ller for scientific discussions. This work was partially conceived during a stay of BA, SL and DL at Centro Internazionale per la Ricerca Matematica (CIRM) in Trento, in the framework of the ``Research in Pairs'' program; all authors gratefully acknowledge financial and logistic support by CIRM.
SL acknowledges financial support by the European Union (ERC \textsc{FermiMath} nr.~101040991). Views and opinions expressed are those of the authors and do not necessarily reflect those of the European Union or the European Research Council Executive Agency. Neither the European Union nor the granting authority can be held responsible for them.
DL acknowledges financial support by Friedrich-Alexander-Universität Erlangen-Nürnberg through the funding program “Emerging Talent Initiative” (ETI), and was partially supported by the project TEC-2024/COM-84 QUITEMAD-CM.
JVM acknowledges support by the Ministry
 of Culture and Science of the State of North Rhine-Westphalia within the project ‘PhoQC’ (Grant Nr.
 PROFILNRW-2020-067). The authors were further partially supported by Gruppo Nazionale per la Fisica Matematica (GNFM) in Italy.

\section{Setting and main results}\label{sec:setting}

\subsection{Mathematical setup} 

We begin by briefly recalling the construction of GSB models in the absence of UV divergences. This discussion is essentially analogous to the one in~\cite{lonigro2022generalized,lill2025self}. Recall that the finite-dimensional spin system is described by $ \fh \simeq \mathbb{C}^D $. As for the bosons, consider a $\sigma$-finite measure space $ (X,\Sigma,\mu) $. The single-boson space is then $ \hilb = L^2(X) $ and the bosonic Fock space is $ \fock $ as defined in Eq.~\eqref{eq:fock}, with vacuum vector $ \Omega := (1,0,0, \ldots ) $. The total Hilbert space of the GSB model is then
$\hfrak=\mathbb{C}^D\otimes\fock$. We denote the $ n $-boson component of $ \psi \in \fock $ and $ \Psi \in \hfrak $ by $ \psi^{(n)} \in S_n L^2(X)^{\otimes n} $ and $ \Psi^{(n)} \in \mathbb{C}^D \otimes  S_n L^2(X)^{\otimes n} $, respectively.

The free evolution of the spin system will be generated by a symmetric operator $K \in\mathbb{C}^{D \times D} $. As for the boson field: given a measurable function $ \omega: X \to \mathbb{R} $ called dispersion relation, the corresponding multiplication operator
\begin{equation}
    \omega:  L^2(X) \supset \dom(\omega) \to  L^2(X) \;, \quad
    \dom(\omega)
    := \left\{ \phi \in  L^2(X) \; : \; \int_X |\phi(k)|^2 \omega(k)^2 \; \di k < \infty \right\}\;,
\end{equation}
is well-defined and self-adjoint. We can then describe the total energy of the boson field by the second quantization of $\omega$, that is, the operator $ \di \Gamma (\omega) : \fock \supset \dom(\di \Gamma (\omega)) \to \fock $, defined as
\begin{equation}
\begin{aligned}
    \left( \di \Gamma (\omega) \psi \right)^{(n)} (k_1, \ldots , k_n)
    := &\sum_{j = 1}^n \omega(k_j) \psi^{(n)} (k_1, \ldots, k_n) \;,\\
    \dom(\di \Gamma(\omega))
    := &\{ \psi \in \fock \; : \; \Vert \di \Gamma(\omega) \psi \Vert_{\fock} < \infty \}
    \subset \fock \;.
\end{aligned}
\end{equation}
Following~\cite[Sect.~VIII.10]{reed1972methods}, $ \di \Gamma(\omega) $ is also self-adjoint.

In order to describe the interaction between the spin system and the boson field, we introduce the creation and annihilation operators $\a{f},\adag{f}$ for $f\in L^2(X)$, acting as
\begin{equation} \label{eq:adaggera}
\begin{aligned}
    a^*(f) S_n (\phi_1 \otimes \ldots \otimes \phi_n) := &\sqrt{n+1} \,S_{n+1}(f \otimes \phi_1 \otimes \ldots \otimes \phi_n),\\ 
    a(f) S_n (\phi_1 \otimes \ldots \otimes \phi_n) := &\frac{1}{\sqrt{n}} S_{n-1}\sum_{j = 1}^n \langle f, \phi_j \rangle (\phi_1 \otimes \ldots \otimes \phi_{j-1} \otimes \phi_{j+1} \otimes \ldots \otimes \phi_n)
\end{aligned}
\end{equation}
for any $ \phi_1, \ldots, \phi_n \in  L^2(X) $.
It is well-known (see also Lemma~\ref{lem:aastar_bounds_matrix}) that, by linearity, these relations uniquely define operators on $ \dom(\mathcal{N}^{\frac 12}) \subset \fock $, where the operator $\mathcal{N}: \fock \supset \dom(\mathcal{N}) \to \fock \;$ is the self-adjoint number operator defined by
\begin{equation} \label{eq:cN}
        (\mathcal{N} \psi)^{(n)} := n \psi^{(n)} \;, \quad
    \dom(\mathcal{N}) := \left\{ \psi \in \fock : \sum_{n=0}^\infty n^2 \Vert \psi^{(n)} \Vert^2 < \infty \right\} \;,
\end{equation}
 and $ \mathcal{N}^{\frac 12} $ is defined via spectral calculus. Furthermore, they satisfy the canonical commutation relations (CCR): for all $ f,g  \in  L^2(X) \;$ and $\psi\in\dom(\mathcal{N})$,
    \begin{equation}
        [a(f), a^*(g)]\psi = \langle f, g \rangle \psi\;, \qquad 
        [a(f), a(g)]\psi = [a^*(f), a^*(g)]\psi = 0 \;.
    \end{equation}
We also note that $\a{f},\adag{f}$ are formally mutually adjoint, and they are relatively bounded with respect to $\dOmega$ with relative bound equal to zero. Consequently, in ``Case 0'', where $N\in\mathbb{N}$, $f_1,\dots,f_N\in L^2(X)$, and $B_1,\dots,B_N\in\mathbb{C}^{D \times D}$, we can define a renormalized Hamiltonian with domain $\dom(H)= \mathbb{C}^D \otimes \dom (\di \Gamma (\omega) )$ by
\begin{equation}\label{eq:bare}
    H = H^{\bare}
    = K \otimes 1
        + 1 \otimes \di \Gamma (\omega)
        +\sum_{j=1}^N \left(B_j\otimes\adag{f_j}+B_j^*\otimes\a{f_j}\right) \;.
\end{equation}
By virtue of the Kato--Rellich theorem, $H$ is self-adjoint~\cite[Chapter~1~and~5]{arai2018analysis}. This is a generalized spin--boson (GSB) model without UV divergences.

To simplify notation, in what follows we will denote operators of the form $ 1 \otimes A $ or $ A \otimes 1 $ simply as $ A $, whenever there is no risk of confusion on which tensor factor $ A $ is acting.
    
\subsection{Renormalization of GSB models of Case 2}

The expression~\eqref{eq:bare} ceases to be well-defined, in general, whenever $f_j\notin L^2(X)$ at least for some $j=1,\dots,N$. While $\a{f_j}$ is then still densely defined, it does not admit a densely defined adjoint~\cite{nelson1964interaction}. Recalling the classification of divergences given in Section~\ref{sec:introduction}, if all form factors belong to Case 1, i.e.,
\begin{equation}\label{eq:mild}
    \omega^{-1/2}f_j\in L^2(X) \qquad \forall j=1,\dots,N \;,
\end{equation}
then the interaction term in Eq.~\eqref{eq:bare} can be rigorously interpreted as a \textit{form} perturbation of $\dOmega$, again with relative bound zero; whence there exists a unique self-adjoint operator $H=H^{\bare}$ associated with Eq.~\eqref{eq:bare}~\cite{lonigro2022generalized}, whose domain can be explicitly characterized under certain conditions on the operators $B_1,\dots,B_N$, cf.~\cite{lill2025self}. Besides, given any ultraviolet regularization of the form factors, i.e. $(f_{1,\Lambda})_{\Lambda\geq0},\dots,(f_{N,\Lambda})_{\Lambda\geq0}$ such that 
\begin{equation}
    \omega^{-1/2}f_{j,\Lambda}\to\omega^{-1/2}f_j \qquad \forall j=1,\dots,N \;,
\end{equation}
the corresponding GSB operator $H_\Lambda$ (see~\eqref{eq:HLambda}) converges to $H$ in the norm resolvent sense---regardless of the particular choice of regularization. 

\bigskip

In this article, we consider the more general Case 2.

\begin{assumption} \label{as:strong}
    The dispersion relation $\omega:X\to\mathbb{R}$ is measurable and has a positive boson mass $ m > 0 $, i.e., $ \omega \ge m $ almost everywhere. Furthermore, the form factors $ f_j: X \to \mathbb{C} $ are measurable and satisfy
    \begin{equation}
       \Vert \omega^{-1} f_j \Vert < \infty \quad \forall j\in \{1, \ldots, N\} \;,
    \end{equation}
    and we assume 
    \begin{equation}
    \label{CommutingHypo}
        [B_j, B_{j'}] = 0 \quad \forall j, j' \in \{1, \ldots, N\} \;.
    \end{equation}
\end{assumption}

\begin{remark}
Our assumptions, in particular, cover the \emph{UV critical case} where $ \int \omega^{-2} |f_j|^2 < \infty $, while $ \int \omega^{-2+\varepsilon} |f_j|^2 = \infty $ for any $ \varepsilon > 0 $. This latter case was also covered in~\cite{hinrichs2025ultraviolet}, but with much stricter assumptions on $ B_j $ which enabled renormalization via the interior-boundary condition (IBC) method with one single iteration step. For generic $ B_j $, we expect both an IBC renormalization with an arbitrarily large number of iteration steps as in~\cite{lampart2023hamiltonians}, and the iterative expansion of~\cite{alvarez2023ultraviolet} to succeed. However, as in~\cite{lampart2023hamiltonians,alvarez2023ultraviolet}, both methods will likely not cover the critical case, that is, we expect that they only work if $ \int \omega^{-2+\varepsilon} |f_j|^2 < \infty $. Moreover, our renormalization method does not require an iterative procedure and is therefore technically much simpler than~\cite{lampart2023hamiltonians,alvarez2023ultraviolet}.
\end{remark}

\begin{remark}
We believe that the assumption $ \omega \ge m > 0 $ may also be relaxed to ``infrared divergences of Case 1'' in the sense that $ \Vert \chi(\omega < m) \omega^{-1/2} f \Vert < \infty $, as in~\cite{hinrichs2025ultraviolet}. However, to keep the discussion simple, we do not discuss infrared divergences, here.
\end{remark}

Our procedure to obtain a self-adjoint operator from the formal expression $H^{\rm bare}$ is rooted in the following heuristic observation. We begin by formally subtracting an infinite \emph{self-energy} (also cf.~Remark~\ref{rem:infiniteself}): 
\begin{equation}
    E_\infty := - \sum_{j, j' = 1}^N B_j^* B_{j'} \int \frac{\overline{f_j(k)} f_{j'}(k)}{\omega(k)} \,\di k \;.
\end{equation}
The renormalized Hamiltonian will formally read $ H = H^{\bare} - E_\infty $, and its construction is based on the following formal observation (cf.~\cite{glimm1968lambda,glimm1970lambda,glimm1970lambdaIII}):
\begin{equation} \label{eq:completingsquare_T}
    H^{\bare} - K - E_\infty
    = \int \omega(k) \hat a(k)^* \hat a(k)\, \di k \;,
\end{equation}
where we introduced the operator-valued distribution for $ k \in X $
\begin{equation}
\hat a(k) := a(k) + \sum_{j=1}^N B_j \frac{f_j(k)}{\omega(k)} \;. 
\end{equation}
Since $ K $ is bounded, we can eventually add it without affecting self-adjointness. 

We will make mathematical sense of the r.h.s. of~\eqref{eq:completingsquare_T} in the following way: We consider an orthonormal basis $ (e_\ell)_{\ell \in \mathbb{N}} \subset \dom(\omega) $ of $L^2(X)$. This choice ensures $\Vert \omega e_\ell \Vert < \infty$ and thus $ |\langle e_\ell, f_j \rangle| \le \Vert \omega e_\ell \Vert \Vert \omega^{-1} f_j \Vert < \infty $. We then define the operators
\begin{equation} \label{eq:a_ell}
    a^*_\ell := a^*(e_\ell) \;, \qquad
    a_\ell := a(e_\ell) \;,
\end{equation}
which satisfy the standard CCR: $ [a_\ell, a_{\ell'}^*] = \delta_{\ell, \ell'} $ and $ [a_\ell, a_{\ell'}] = [a_\ell^*, a_{\ell'}^*] = 0 $. The displaced annihilation operator is then defined for $ g \in L^2(X) $ as
\begin{equation} \label{eq:ahat_ell_T}
    \hat a(g) := a(g) + \sum_{j = 1}^N B_j \Braket{ g, \frac{f_j}{\omega}} \;, \qquad
    \hat a_\ell := \hat a(e_\ell) \;,
\end{equation}
and the basis expansion of~\eqref{eq:completingsquare_T} reads
\begin{equation} \label{eq:completingsquare_ell}
    H^{\bare} - K - E_\infty
    = \sum_{\ell, \ell' \in \mathbb{N}} \langle e_\ell, \omega e_{\ell'} \rangle \hat a_{\ell}^* \hat a_{\ell'} \;.
\end{equation}

The key idea of our renormalization procedure is to interpret the r.h.s. as an operator on specific dressed test vectors $ \Psi \in \hfrak $, which are built upon some dressed vacuum vector on which $ \hat{a}_\ell $ vanishes. To generate the dressed vacuum, we define the following dressing transformation:
   	
\begin{definition} \label{def:T}
Let $ B: \hfrak \supset \dom(B) \to \hfrak $ be defined by
\begin{equation} \label{eq:B}
   	B := -\sum_{j=1}^N B_j  a^* \left( \frac{f_j}{\omega} \right) \;, \qquad
   	\dom(B) := \left\{ \Psi \in \hfrak : \sum_{m=0}^\infty \Vert (B \Psi)^{(m)} \Vert^2 < \infty \right\} \;.
\end{equation}
Then the \textbf{dressing operator} $ T : \hfrak \supset \dom(T) \to \hfrak $ is defined by 
\begin{equation} \label{eq:T}
   	T := \e^{B} = \sum_{n = 0}^\infty \frac{B^n}{n!} \;, \qquad
   	\dom(T) := \left\{ \Psi \in \hfrak : \sum_{m=0}^\infty \Vert (T \Psi)^{(m)} \Vert^2 < \infty \right\} \;.
\end{equation}
\end{definition}
Note that each sector $ (B \Psi)^{(m)} $ is well-defined for any $ \Psi \in \hfrak $ and $ m \in \mathbb{N}_0 $. Similarly, each sector $ (T \Psi)^{(m)} = \sum_{n = 0}^\infty \frac{1}{n!} (B^n \Psi)^{(m)}$ is well-defined for any $ \Psi \in \hfrak $ and $ m \in \mathbb{N}_0 $, as it only contains contributions from finitely many $ n $.  
The main motivation for this definition is the following: all dressed annihilation operators $ \hat{a}_\ell $~\eqref{eq:ahat_ell_T} vanish on any vector in the form $ T (v \otimes \Omega) $, which can therefore be interpreted as a ``dressed vacuum''. We prove this in Lemma~\ref{lem:T_properties}.

Our main result about the construction of the renormalized Hamiltonian $ H = H^{\bare} - E_\infty $ is the following:
    
\begin{theorem}[Renormalization in Case 2] \label{thm:HE2def_T}
Let $ \omega $, $ f_j $, and $ B_j $ satisfy Assumption~\ref{as:strong}. Then, for $ n \in \mathbb{N}_0 $, $ \underline{\ell} = (\ell_1, \ldots, \ell_n) \in \mathbb{N}^n $, and $ v \in \mathbb{C}^d $, the dressed test vectors
\begin{equation} \label{eq:Psi_ell}
    \Psi_{\underline{\ell}}
    := a^*_{\ell_1} \ldots a^*_{\ell_n} T (v \otimes \Omega) \in \hfrak \;,
\end{equation}
with $T$ as in Eq.~\eqref{eq:T}, are well-defined. Furthermore, the test domain
\begin{equation} \label{eq:cD}
    \mathcal{D} := \mathrm{Span}\{ a^*_{\ell_1} \ldots a^*_{\ell_n} T (v \otimes \Omega) \in \hfrak : n \in \mathbb{N}_0, \; \underline{\ell} \in \mathbb{N}^n , \; v \in \mathbb{C}^D \}
\end{equation}
is dense and the operator $ H: \mathcal{D} \to \hfrak $ acting as
\begin{equation} \label{eq:widetilde_H}
    H = \sum_{\ell, \ell' \in \mathbb{N}} \langle e_\ell, \omega e_{\ell'} \rangle \hat{a}_\ell^* \hat{a}_{\ell'} + K \;,
\end{equation}
with $\hat{a}_{\ell}$ as per Eq.~\eqref{eq:ahat_ell_T}, is well-defined and semi-bounded from below. Therefore, it has a self-adjoint Friedrichs extension, which we also call $ H $.
\end{theorem}

We prove this theorem in Section~\ref{sec:proof_renormalization_weylfree}. 

In order to properly interpret the operator $H$ given by Theorem~\ref{thm:HE2def_T} as the desired renormalization of the GSB Hamiltonian, we need to check if $H$ can be indeed reconstructed from cut-off renormalization. We answer this question affirmatively:

\begin{theorem}[Norm resolvent convergence of cut-off Hamiltonians] \label{thm:normresolvent}
Let $ \omega $, $ f_j $, and $ B_j $ satisfy Assumption~\ref{as:strong}. For each $ 1 \le j \le N $, consider a sequence $ (f_{j,\Lambda})_{\Lambda \in \mathbb{N}} \subset L^2(X) $ of cut-off form factors satisfying $ \frac{f_{j,\Lambda}}{\omega} \to \frac{f_j}{\omega} $. Then, defining
\begin{equation} \label{eq:HLambda}
    H_\Lambda := K + \di \Gamma (\omega)
        + \sum_{j=1}^N (B_j a^*(f_{j,\Lambda}) + B_j^* a(f_{j,\Lambda}) ) \;, \quad
    E_\Lambda := - \sum_{j,j'=1}^N B_j^* B_{j'} \int \frac{\overline{f_{j,\Lambda}(k)} f_{j',\Lambda}(k)}{\omega(k)} \di k \;,
\end{equation}
the cut-off Hamiltonians $ (H_\Lambda - E_\Lambda) $ converge to the operator $ H $, constructed in Theorem~\ref{thm:HE2def_T}, in the norm resolvent sense as $ \Lambda \to \infty $.
\end{theorem}

We prove this theorem in Section~\ref{sec:proof_normresolvent}. Some comments about our two results are in order:

\begin{remark}
    Following the same proof steps as for Theorem~\ref{thm:normresolvent}, one easily sees that for $ \Re(-z) $ large enough, $ ((H_\Lambda-E_\Lambda-z)^{-1})_{\Lambda \in \mathbb{N}} $ is a Cauchy sequence in operator norm topology and thus converges to some bounded $ R(z) $. By a resolvent reconstruction as in~\cite{alvarez2021ultraviolet,alvarez2023ultraviolet}, we could then reconstruct a self-adjoint Hamiltonian $ \tilde{H} $, whose resolvent is $ R(z) $. By Theorem~\ref{thm:normresolvent}, the limit $ R(z) $ is also the resolvent of $ H $, so in particular, $ \tilde{H} = H $. That is, both cut-off renormalization and our direct construction through Friedrichs extension yield the same renormalized Hamiltonian $ H $.
\end{remark}

\begin{remark}\label{rem:infiniteself}
While the construction presented in Theorem~\ref{thm:HE2def_T} gives rigorous meaning to the right-hand side of Eq.~\eqref{eq:completingsquare_T}, the left-hand side---which involves an ``infinite energy''---could also be given a rigorous mathematical interpretation by using the Extended Space State formalism~\cite{lill2022time}: one interprets the infinite term as an element of some vector space that extends $ \mathbb{C}$.
\end{remark}

\subsection{Triviality of certain GSB models in Case 3}

Recently, Dam and M\o ller proved~\cite{dam2020asymptotics} that the standard spin--boson model ($ D = 2, N = 1, K = \eta \sigma_z, B = \lambda \sigma_x $, with $ \eta, \lambda > 0 $) becomes \textit{trivial} if the form factor has a UV-singularity of Case 3, that is, 
\begin{equation}
    \int \frac{|f(k)|^2}{\omega(k)^2} \chi_{\omega \ge c}(k)\;\di k = \infty 
\end{equation}
for some $ c > 0 $. Their results also cover the massless case where $ \inf_k \omega(k) = 0 $, provided $ \int\omega^{-1}(k)|f(k)|^2 < \infty $ (``IR Case 1''). Here, ``triviality'' means that there exist $ E_\Lambda \in \mathbb{R} $ and $ W_\Lambda : \hfrak \to \hfrak $ unitary, such that $ W_\Lambda (H_\Lambda - E_\Lambda) W_\Lambda^* \to \di \Gamma(\omega) $ in the norm resolvent sense.

We now provide a simple condition for GSB models with $N=1$, which allows for a straightforward generalization of the triviality result in~\cite{dam2020asymptotics}. As $N=1$, we simply call $B_1 =: B$.
    
\begin{assumption} \label{as:eigenbasis}
$ B $ is unitary and permutes some eigenbasis of $ K $. That is, there exists an orthonormal eigenbasis $ (v_k)_{k = 1}^D \subset \mathbb{C}^D $ of $ K $ such that for all $ 1 \le k \le D $, there is some $ 1 \le k' \le D $ with $ B v_k = v_{k'} $.
\end{assumption}

\begin{proposition}[Triviality of specific GSB models in Case 3] \label{prop:triviality}
Consider the case $N = 1$ and suppose Assumption~\ref{as:eigenbasis} holds. Let $ (f_\Lambda)_{\Lambda \in \mathbb{N}} $ be a family of cut-off form factors with $ \Vert \omega^{-1} f_\Lambda \Vert \to \infty $ as $ \Lambda \to \infty $ and recall the definition~\eqref{eq:HLambda} of $ H_\Lambda $ and $ E_\Lambda $. Then there exists a diagonal matrix $ \eta_0 \in \mathbb{C}^{D \times D} $, and a unitary operator $ W_\Lambda : \hfrak \to \hfrak $, such that
\begin{equation}
    W_\Lambda (H_\Lambda - E_\Lambda) W_\Lambda^*
    \to \eta_0 + \di \Gamma(\omega) \qquad \textnormal{as } \Lambda \to \infty
\end{equation}
in the norm resolvent sense.
\end{proposition}

We provide the proof of this proposition in Appendix~\ref{sec:triviality}. Note that, physically, the diagonal matrix $ \eta_0 $ just describes the renormalized energy levels.

\begin{remark}
The argument for proving Proposition~\ref{prop:triviality} strongly depends on the existence of a fiber decomposition of the Hamiltonian, with each fiber of the Hamiltonian taking a specific form in terms of generalized Weyl operators. This fiber decomposition only exists for very specific $ K $ and $ B $, and there is no reason to believe that triviality holds in general for GSB models with $ \int \omega^{-2} |f_j|^2 = \infty $, as confirmed by the recent results in~\cite{falconi2025nontrivial}.
\end{remark}

\section{Proof of renormalization for UV-divergences of Case 2}
\label{sec:proof_renormalization_weylfree}
    
For the proof of Theorem~\ref{thm:HE2def_T}, let us first generalize a well-known bound on the action of annihilation and creation operators:

\begin{lemma} \label{lem:aastar_bounds_matrix}
Let $ \mathfrak{h} $ be a separable Hilbert space, and  $ M: \mathbb{N} \to \mathcal{B}(\mathfrak{h}) $ a function satisfying
\begin{equation}
    \Vert M \Vert_2
    := \Big( \sum_{\ell \in \mathbb{N}} \Vert M_\ell \Vert_{\mathcal{B}(\mathfrak{h})}^2 \Big)^{\frac 12} < \infty \;,
\end{equation}
 with $ \mathcal{B}(\mathfrak{h}) $ being the space of bounded operators on $\mathfrak{h}$. Then, for any $ \Psi \in \mathfrak{h} \otimes \dom(\mathcal{N}^{\frac 12}) $, the following bounds hold:
\begin{equation}
    \Big\Vert \sum_{\ell \in \mathbb{N}} (M_\ell \otimes a_\ell^*) \Psi \Big\Vert
    \le \Vert M \Vert_2 \Vert (\mathcal{N} + 1)^{\frac 12} \Psi \Vert \;, \qquad
    \Big\Vert \sum_{\ell \in \mathbb{N}} (M_\ell \otimes a_\ell) \Psi \Big\Vert
    \le \Vert M \Vert_2 \Vert \mathcal{N}^{\frac 12} \Psi \Vert \;,
\end{equation}
with $a_\ell,a_\ell^*$ as per Eq.~\eqref{eq:a_ell}.
\end{lemma}
\begin{proof} 
For annihilation operators, using the Cauchy--Schwarz inequality and the equality $ \mathcal{N} = \sum_{\ell} a_\ell^* a_\ell $,
\begin{equation}
\begin{aligned}
    &\Big\Vert \sum_{\ell \in \mathbb{N}} (M_\ell \otimes a_\ell) \Psi \Big\Vert^2
    \le \sum_{\ell, \ell' \in \mathbb{N}} \Vert M_\ell \Vert_{\mathcal{B}(\mathfrak{h})} \Vert M_{\ell'} \Vert_{\mathcal{B}(\mathfrak{h})}
        \Vert a_{\ell'} \Psi \Vert \Vert a_{\ell} \Psi \Vert \\
    &\le \Big( \sum_{\ell, \ell' \in \mathbb{N}} \Vert M_\ell \Vert^2_{\mathcal{B}(\mathfrak{h})} \Vert M_{\ell'} \Vert^2_{\mathcal{B}(\mathfrak{h})} \Big)^{\frac 12}
        \Big( \sum_{\ell, \ell' \in \mathbb{N}} \Vert a_{\ell'} \Psi \Vert^2 \Vert a_{\ell} \Psi \Vert^2 \Big)^{\frac 12}
    \le \Vert M \Vert_2^2 \Vert \mathcal{N}^{\frac 12} \Psi \Vert^2 \;.
\end{aligned}
\end{equation}
For creation operators, the CCR render
\begin{equation}
    \Big\Vert \sum_{\ell \in \mathbb{N}} (M_\ell \otimes a_\ell^*) \Psi \Big\Vert^2
    = \sum_{\ell \in \mathbb{N}} \langle \Psi, M_\ell^* M_\ell \Psi \rangle
        + \Big\Vert \sum_{\ell \in \mathbb{N}} (M_\ell \otimes a_\ell) \Psi \Big\Vert^2 
    \le \Vert M \Vert_2^2 (\Vert \Psi \Vert^2 + \Vert \mathcal{N}^{\frac 12} \Psi \Vert^2) \;,
\end{equation}
thus proving the claim.
\end{proof}

Next, we prove that all \textit{dressed} annihilation operators $ \hat{a}_\ell $~\eqref{eq:ahat_ell_T} annihilate any dressed vacuum, i.e. any vector of the form $ T (v \otimes \Omega) $, where $T$ is the dressing operator (Definition~\ref{def:T}). The proof involves the operator
\begin{equation}
        \dom(e^{-B}) := \left\{ \Psi \in \hfrak : \sum_{m=0}^\infty \Vert (e^{-B} \Psi)^{(m)} \Vert^2 < \infty \right\} \;,
\end{equation}
which we can define sector-wise, just as $ T = e^B $ in~\eqref{eq:T}. Note that $ e^{-B} T = 1 $ does not hold as an operator identity, since $ \dom(T) \neq \dom(1) = \hfrak $, but it holds on a dense subspace of $ \hfrak $.

\begin{lemma}[Properties of the dressing operator $ T $] \label{lem:T_properties}
    Let $ \omega $, $ f_j $, and $ B_j $ satisfy Assumption~\ref{as:strong}.
    Then, we have $ e^{-B} T \Psi = \Psi $ for all $ \Psi \in \dom(T) $, and we have $ T e^{-B} \Psi = \Psi $ for all $ \Psi \in \dom(e^{-B}) $. Here, both $ \dom(T) $ and $ \dom(e^{-B}) $ contain the dense space
    \begin{equation}
        \hfrak_{\mathrm{fin}} := \left\{ \Psi \in \hfrak : \exists n_{\max} \in \mathbb{N} : \Psi^{(n)} = 0 \; \forall n > n_{\max} \right\} \;.
    \end{equation}
    Furthermore, for any $ \ell \in \mathbb{N} $ we have
    \begin{equation} \label{eq:ahat_annihilation_T}
        \hat{a}_\ell T (v \otimes \Omega) = 0 \qquad
        \forall \ell \in \mathbb{N} \;, v \in \mathbb{C}^D \;.
    \end{equation}
\end{lemma}
\begin{proof}
First, note that $ \ran(T) \subset \dom(e^{-B}) $. To see this, we have to explicitly compute every sector $ (e^{-B} T \Psi)^{(n)} $ for $ \Psi \in \dom(T) $. As $ B $ only creates particles, this sector contains a finite sum of terms, so it is well-defined, and by explicit computation one checks $ (e^{-B} T \Psi)^{(n)} = \Psi^{(n)} $. Thus, $ e^{-B} T \Psi $ is well-defined and further $ e^{-B} T \Psi = \Psi $. The same argument gives $ T e^{-B} \Psi = \Psi $ for $ \Psi \in \dom(e^{-B}) $.\\
To see that $ \dom(T) $ contains $ \hfrak_{\mathrm{fin}} $, by linearity it suffices to prove $ \Psi \in \dom(T) $ for $ \Psi $ having exactly $ n \in \mathbb{N} $ particles, that is, $ \mathcal{N} \Psi = n \Psi $. We expand $ T = e^B $, keeping in mind that $ \langle B^m \Psi, B^{m'} \Psi \rangle = 0 $ if $ m \neq m' $, since $ B $~\eqref{eq:B} exactly creates one particle, and we use $ \Vert a^*(f) \Phi \Vert \le \Vert f \Vert \Vert (\mathcal{N} + 1)^{\frac 12} \Phi \Vert $:
\begin{equation}
\begin{aligned}\label{eq:bound-t}
    \Vert T \Psi \Vert^2
    &= \sum_{m=0}^\infty \frac{\Vert B^m \Psi \Vert^2}{(m!)^2}
    \le \sum_{m=0}^\infty \frac{1}{(m!)^2} \Big( N \sup_{1 \le j \le N} \Vert B_j \Vert \Vert \omega^{-1} f_j \Vert \Big)^{2m} \Vert (\mathcal{N}+m)^{\frac 12} \ldots (\mathcal{N}+1)^{\frac 12} \Psi \Vert^2 \\
    &\le \sum_{m=0}^\infty \frac{1}{(m!)^2} \frac{(m+n)!}{n!} \Big( N \sup_{1 \le j \le N} \Vert B_j \Vert \Vert \omega^{-1} f_j \Vert \Big)^{2m} \Vert \Psi \Vert^2
    < \infty \;,
\end{aligned}
\end{equation}
so $ \Psi \in \dom(T) $. The same argument applies to $ e^{-B} $.
For proving~\eqref{eq:ahat_annihilation_T}, recall that $ \hat{a}_\ell $ was defined in~\eqref{eq:ahat_ell_T}. By sector-wise computation, one easily checks $ e^{-B} a_\ell e^B = a_\ell + [a_\ell,B] $ as a strong operator identity on $ \hfrak_{\mathrm{fin}} $. From $ [B_j, B] = 0 $ and $ [a_\ell, B] = - \sum_{j=1}^N B_j \Big\langle e_\ell, \frac{f_j}{\omega} \Big\rangle $, we then get
\begin{equation}
    \hat{a}_\ell T
    = e^B \Big( e^{-B} a_\ell e^B + \sum_{j=1}^N B_j \Big\langle e_\ell, \frac{f_j}{\omega} \Big\rangle \Big)
    = T \Big( a_\ell + [a_\ell,B] + \sum_{j=1}^N B_j \Big\langle e_\ell, \frac{f_j}{\omega} \Big\rangle \Big)
    = T a_\ell
\end{equation}
again as a strong operator identity on $ \hfrak_{\mathrm{fin}} $. Then, $ a_\ell(v \otimes \Omega) = 0 $ implies~\eqref{eq:ahat_annihilation_T}. 
\end{proof}

Well-definedness of the excited dressed vacuum vectors $ \Psi_{\underline{\ell}} $ in~\eqref{eq:Psi_ell} is then a consequence of the following lemma:
    
\begin{lemma}[Particle number decay of the dressed vacuum] \label{lem:bound_with_powers_of_N}
Let $ T $ be the dressing operator as per~\eqref{eq:T}, let $ \omega $, $ f_j $, and $ B_j $ satisfy Assumption~\ref{as:strong}, and let $ v \in \mathbb{C}^D $. Then, for every $ \varepsilon > 0 $, there exists a $ C_\varepsilon > 0 $ such that
\begin{equation} \label{eq:bound_with_powers_of_N}
    \Vert (\mathcal{N} + 1)^{\frac 12} \ldots (\mathcal{N} + n)^{\frac 12} T (v \otimes \Omega) \Vert < C_\varepsilon (n!)^{\frac{1+\varepsilon}{2}} \qquad \forall n \in \mathbb{N}_0 \;,
\end{equation}
where $ \mathcal{N}$ is the number operator~\eqref{eq:cN}. 
\end{lemma}
\begin{proof}
Recalling $ T = e^B $, where $ B $ was defined in~\eqref{eq:B}, and by using Eq.~\eqref{eq:adaggera}, one explicitly computes
\begin{equation} \label{eq:displacedvacuum_T}
    (T (v \otimes \Omega) )^{(m)} = S_m \frac{(-1)^m}{\sqrt{m!}} \sum_{j_1, \ldots, j_m = 1}^N
        (B_{j_1} \ldots B_{j_m} v) \otimes 
        \Big( \frac{f_{j_1}}{\omega} \otimes \ldots \otimes \frac{f_{j_m}}{\omega} \Big) 
\end{equation}
for $ m \ge 1 $, and $ (T (v \otimes \Omega) )^{(0)} = v \otimes \Omega $. By using Stirling's formula\footnote{In this section, $ C $ denotes a constant independent of $ n $ and $ m $, which may change from line to line.},
\begin{equation} \label{eq:Stirling}
    n! \ge C n^n \e^{-n} \sqrt{n} \qquad \Rightarrow \qquad 
    n^n \le C_\varepsilon (n!)^{1 + \varepsilon} \;,
\end{equation}
we get, uniformly in $ n $,
\begin{equation}
\begin{aligned}
	&\quad \Vert (\mathcal{N}+1)^{\frac 12} \ldots (\mathcal{N}+n)^{\frac 12} T (v \otimes \Omega) \Vert^2
    = \sum_{m=0}^\infty \Vert (\mathcal{N}+1)^{\frac 12} \ldots (\mathcal{N}+n)^{\frac 12} (T (v \otimes \Omega) )^{(m)} \Vert^2 \\
	&\le \sum_{m=0}^\infty \frac{(m+n)!}{(m!)^2}
        \left( N \sup_{1 \le j \le N} \Vert B_j \Vert \Big\Vert \omega^{-1} f_j \Big\Vert \right)^{2m}
	\le \sum_{m=0}^\infty C^m \frac{(m+n)!}{(m!)^2} \\
	&\le \sum_{m=0}^{2n} C^m (3n)^{n}
		+ \sum_{m = 2n+1}^\infty C^m \frac{(m+1) \ldots (m+n)}{m!} \\
	&\le C_\varepsilon (n!)^{1 + \varepsilon}
		+ \sum_{m = 2n+1}^\infty C^m 2^{n} \frac{(m-n) \ldots m}{m!}
	= C_\varepsilon (n!)^{1 + \varepsilon}
		+ (2C)^n \sum_{m = n}^\infty \frac{C^m}{m!} \\
	&\le C_\varepsilon (n!)^{1 + \varepsilon} \;,
\end{aligned}
\end{equation}
where we applied Eq.~\eqref{eq:bound-t} in the second step. This is the claimed inequality.
\end{proof}

Next, we show that the test domain $\mathcal{D}$ used in the construction of our renormalized operator in Theorem~\ref{thm:HE2def_T} is indeed dense:

\begin{proposition}[Denseness of $ \mathcal{D} $] \label{prop:D_dense}
Let $ \omega $, $ f_j $, and $ B_j $ satisfy Assumption~\ref{as:strong}. Then the space $\mathcal{D}\subset\hfrak$ defined in Eq.~\eqref{eq:cD} is dense.
\end{proposition}
\begin{proof}
It is well-known that the span of vectors of the form
\begin{equation} \label{eq:target_vector}
    \Psi = a^*_{\ell^{(1)}} \ldots a^*_{\ell^{(k)}} (v \otimes \Omega) \;, \qquad \textnormal{with } \; k \in \mathbb{N}_0 , \; (\ell^{(1)}, \ldots, \ell^{(k)}) \in \mathbb{N}^k , \; v \in \mathbb{C}^D 
\end{equation}
is dense in $ \hfrak $. So it suffices to show that any such $ \Psi $ can be approximated by vectors from $ \mathcal{D} $. To do so, we proceed as follows: 
\begin{itemize}
    \item[(i)] First, we construct a sequence $ (\tilde\Psi_m)_{m \in \mathbb{N}} \subset \hfrak $ of excited vectors
\begin{equation}
    \tilde\Psi_m \in \mathrm{Span} \{ a^*(g_1) \ldots a^*(g_n) T (v_m \otimes \Omega) \; : \; n \in \mathbb{N}_0 , \; g_1, \ldots, g_n \in L^2(X), \; v_m \in \mathbb{C}^D \} \;
\end{equation}
such that
\begin{equation} \label{eq:vector_approximation_1}
    \lim_{m \to \infty} \Vert \Psi - a^*_{\ell^{(1)}} \ldots a^*_{\ell^{(k)}} \tilde\Psi_m \Vert = 0 \;.
\end{equation}
\item[(ii)] Then, we show that each $ \tilde\Psi_m $ can be approximated by a vector $ \Psi_{m,m'} \in \mathcal{D} $, in the following sense:
\begin{equation} \label{eq:vector_approximation_2}
    \lim_{m' \to \infty} \Vert a^*_{\ell^{(1)}} \ldots a^*_{\ell^{(k)}} (\tilde\Psi_{m} - \Psi_{m,m'}) \Vert = 0 \;.
\end{equation}
\end{itemize}
Combining~\eqref{eq:vector_approximation_1} and~\eqref{eq:vector_approximation_2} then will yield $ a^*_{\ell^{(1)}} \ldots a^*_{\ell^{(k)}} \Psi_{m,m'(m)} \to \Psi $ for a suitable sequence $ m \mapsto m'(m) $, where $ a^*_{\ell^{(1)}} \ldots a^*_{\ell^{(k)}} \Psi_{m,m'(m)} \in \mathcal{D} $, which will conclude the proof.\medskip
    
\textbf{(i)}: To construct $ \tilde\Psi_m $, recall $ T = \e^B $, with $ B $ defined in~\eqref{eq:B}. We thus expect the vector $v \otimes \Omega $ to be well-approximated by a truncation of the exponential series of $ T \e^{-B} (v \otimes \Omega) $, which motivates the following definition:
\begin{equation} \label{eq:Psi_m_final}
    \tilde\Psi_m := \sum_{n=0}^m \frac{\Phi_n}{n!} \;, \qquad \Phi_n := T (-B)^n (v \otimes \Omega)
    = \sum_{j_1, \ldots, j_n = 1}^N
		a^*\Big( \frac{f_{j_1}}{\omega} \Big) \ldots
		a^*\Big( \frac{f_{j_n}}{\omega} \Big)
		T (B_{j_1} \ldots B_{j_n} v \otimes \Omega) \;.
\end{equation}
Using the strong operator identity $ e^B e^{-B} = 1 $ from Lemma~\ref{lem:T_properties}, one readily sees $ \tilde\Psi_m \to v \otimes \Omega $ as $ m \to \infty $. Furthermore, for $ m_1 > m_2 $,
\begin{equation}
\begin{aligned}
    &\|a^*_{\ell^{(1)}} \ldots a^*_{\ell^{(k)}} (\tilde\Psi_{m_1}-\tilde\Psi_{m_2})\|
    \leq \sum_{m = m_2}^\infty \frac{1}{m!}\|a^*_{\ell^{(1)}} \ldots a^*_{\ell^{(k)}} T (-B)^m (v \otimes \Omega)\|\\
    &\leq \sum_{m = m_2}^\infty \frac{1}{m!}\left( N \sup_{1 \le j \le N} \Vert B_j \Vert \Big\Vert \omega^{-1} f_j \Big\Vert \right)^m \Vert (\mathcal{N}+1)^{\frac 12} \ldots (\mathcal{N}+m+k)^{\frac 12} T (v \otimes \Omega) \Vert \\
	&\leq \sum_{m = m_2}^\infty \frac{1}{m!} C_\varepsilon C^m ((m+k)!)^{\frac{1 + \varepsilon}{2}} \to 0 \qquad \textnormal{as } m_2 \to \infty\;,
\end{aligned}
\end{equation}
so $ (\tilde\Psi_m)_{m \in \mathbb{N}} $ is a Cauchy sequence in the graph norm of $ a^*_{\ell^{(1)}} \ldots a^*_{\ell^{(k)}} $ and by the well-known closedness of this operator, we conclude~\eqref{eq:vector_approximation_1}.\medskip

\textbf{(ii)} To establish the second approximation step~\eqref{eq:vector_approximation_2}, it suffices to consider any
\begin{equation}
    \tilde\Psi
    = a^*(g_1) \ldots a^*(g_n) T (w \otimes \Omega) \;,
    \qquad n \in \mathbb{N}_0 , \; g_1, \ldots, g_n \in L^2(X) , \; w \in \mathbb{C}^D \;,
\end{equation}
and to approximate it by a sequence $ \Psi_{m'} \in \mathcal{D} $, in the following sense:
\begin{equation} \label{eq:vector_approximation_2_bis}
    \lim_{m' \to 0} \Vert a^*_{\ell^{(1)}} \ldots a^*_{\ell^{(k)}} (\tilde\Psi - \Psi_{m'}) \Vert = 0 \;.
\end{equation}
To this end, we can naturally approximate $ \tilde\Psi $ by using the basis expansion $ g_{j,m'} := \sum_{\ell = 1}^{m'} \langle e_\ell, g_j \rangle e_\ell $ with $ g_j = \lim_{m' \to \infty} g_{j,m'} $, where the approximating vectors are
\begin{equation}
    \Psi_{m'} := a^*(g_{1,m'}) \ldots a^*(g_{n,m'}) T (w \otimes \Omega) \in \mathcal{D} \;.
\end{equation}
Using $ \Vert g_{j,m'} \Vert \le \Vert g_j \Vert $, the approximation error is then
\begin{equation}
\begin{aligned}
    &\quad \Vert a^*_{\ell^{(1)}} \ldots a^*_{\ell^{(k)}}(\tilde\Psi - \Psi_{m'}) \Vert \\
    &= \Big\Vert \sum_{j=1}^n a^*_{\ell^{(1)}} \ldots a^*_{\ell^{(k)}} a^*(g_1) \ldots a^*(g_{j-1}) a^*(g_j - g_{j,m'}) a^*(g_{j+1,m'}) \ldots a^*(g_{n,m'}) T (w \otimes \Omega) \Big\Vert \\
    &\le \sum_{j=1}^n \Vert g_j - g_{j,m'} \Vert \;
    \Vert g_1 \Vert \ldots \Vert g_{j-1} \Vert \Vert g_{j+1} \Vert \ldots \Vert g_n \Vert
    \underbrace{\Big\Vert (\mathcal{N}+1)^{\frac 12} \ldots (\mathcal{N}+n+k)^{\frac 12} T (w \otimes \Omega) \Big\Vert}_{< \infty \textnormal{ by Lemma~\ref{lem:bound_with_powers_of_N}}} \\
    &\le C \sum_{j=1}^n \Vert g_j - g_{j,m'} \Vert \;. \\
\end{aligned}
\end{equation}
Then, $ \lim_{m' \to \infty} \Vert g_j - g_{j,m'} \Vert = 0 $, so indeed~\eqref{eq:vector_approximation_2_bis} holds, which concludes the proof.
\end{proof}	

At this point, we proved that $ \Psi_{\underline{\ell}} $ in~\eqref{eq:Psi_ell} is well-defined and $ \mathcal{D} $ in~\eqref{eq:cD} is dense. It remains to prove well-definedness of $ (H-K) = \sum_{\ell, \ell'} \langle e_\ell, \omega e_{\ell'} \rangle \hat{a}_\ell^* \hat{a}_{\ell'} $ on $ \mathcal{D} $. In fact, we prove the following more general statement, which will be useful later.

\begin{lemma}[Shifted second quantization operators] \label{lem:di_hat_Gamma}
Let $ \omega $, $ f_j $, and $ B_j $ satisfy Assumption~\ref{as:strong}. Let $ \xi : X \to [0, \infty)$ be measurable and\footnote{Note that $ \dom(\xi) \cap \dom(\omega) $ is dense in $ L^2(X) $, as it contains all measurable functions $ f: X \to \mathbb{C} $ supported on any joint level set $ L_n = \{ k \in X : \omega(k), \xi(k) \le n \} $, where $ L_n \nearrow X $ as $ n \to \infty $.} $ (e_\ell)_{\ell \in \mathbb{N}} \subset \dom(\xi) \cap \dom(\omega) $ an orthonormal basis of $ L^2(X) $. Then, the following operator is densely defined on $ \mathcal{D} $~\eqref{eq:cD} and non-negative:
\begin{equation} \label{eq:di_hat_Gamma}
    \di \hat{\Gamma} (\xi)
    := \sum_{\ell, \ell' \in \mathbb{N}} \langle e_\ell, \xi e_{\ell'} \rangle \hat{a}_\ell^* \hat{a}_{\ell'} \;.
\end{equation}
\end{lemma}

\begin{proof}
By definition~\eqref{eq:ahat_ell_T} of $ \hat{a}_\ell $, we have
\begin{equation} \label{eq:ahat_CCR}
    [\hat{a}_\ell, \hat{a}_{\ell'}^*]
    = \delta_{\ell, \ell'} + \sum_{j, j' = 1}^N [B_j, B_{j'}^*] \Big\langle e_\ell, \frac{f_j}{\omega} \Big\rangle \Big\langle \frac{f_{j'}}{\omega}, e_{\ell'} \Big\rangle \;.
\end{equation}
We now take any $ n \in \mathbb{N}_0 $ and $ \underline{\ell} = (\ell_1, \ldots, \ell_n) \in \mathbb{N}^n $ and compute the action of $\di\hat{\Gamma}(\xi)$ on the corresponding vector $ \Psi_{\underline{\ell}} $ as given in~\eqref{eq:Psi_ell}:
\begin{equation} \label{eq:split_I_II_III}
\begin{aligned}
    \Vert \di \hat{\Gamma} (\xi) \Psi_{\underline{\ell}} \Vert^2
    &= \I + \II + \III \;, \\
    \I &:= \sum_{\ell_1', \ell_2', \ell_3', \ell_4' \in \mathbb{N}_0}
        \langle e_{\ell_1'}, \xi e_{\ell_2'} \rangle 
        \langle e_{\ell_3'}, \xi e_{\ell_4'} \rangle 
        \langle \Psi_{\underline{\ell}}, \hat{a}^*_{\ell_1'} \hat{a}^*_{\ell_3'} \hat{a}_{\ell_2'} \hat{a}_{\ell_4'} \Psi_{\underline{\ell}} \rangle \;, \\
    \II &:= \sum_{\ell_1', \ell_4' \in \mathbb{N}_0}
        \langle e_{\ell_1'}, \xi^2 e_{\ell_4'} \rangle
        \langle \Psi_{\underline{\ell}}, \hat{a}^*_{\ell_1'} \hat{a}_{\ell_4'} \Psi_{\underline{\ell}} \rangle \;, \\
    \III &:= \sum_{\ell_1', \ell_4' \in \mathbb{N}_0}
        \sum_{j, j' = 1}^N
        \Big\langle \xi e_{\ell_1'}, \frac{f_j}{\omega} \Big\rangle
        \Big\langle \frac{f_{j'}}{\omega}, \xi e_{\ell_4'} \Big\rangle 
        \langle \Psi_{\underline{\ell}}, \hat{a}^*_{\ell_1'}[B_j, B_{j'}^*] \hat{a}_{\ell_4'} \Psi_{\underline{\ell}} \rangle \;.
\end{aligned}
\end{equation}
Using $ [\hat{a}_{\ell}, a_{\ell'}^*] = \delta_{\ell, \ell'} $ and Lemma~\ref{lem:T_properties}, we obtain
\begin{equation}
    \hat{a}_{\ell} \Psi_{\underline{\ell}}
    = \sum_{i = 1}^n \delta_{\ell, \ell_i} \Psi_{\underline{\ell}^{(i)}} \;, \qquad \textnormal{where} \qquad
    \Psi_{\underline{\ell}^{(i)}}
    := a^*_{\ell_1} \ldots a^*_{\ell_{i-1}} a^*_{\ell_{i+1}} \ldots a^*_{\ell_n} T (v \otimes \Omega) \;,
\end{equation}
that is, each term $ \Psi_{\underline{\ell}^{(i)}} $ is obtained by removing $ a^*_{\ell_i} $ from $\Psi_{\underline{\ell}}$. Now,
\begin{equation}
\begin{aligned}
    |\II|
    &= \bigg\vert \sum_{\ell_1', \ell_4' \in \mathbb{N}_0}
        \langle e_{\ell_1'}, \xi^2 e_{\ell_4'} \rangle
        \sum_{i_1, i_4 = 1}^n
        \delta_{\ell_1',\ell_{i_1}} \delta_{\ell_4',\ell_{i_4}} \langle \Psi_{\underline{\ell}^{(i_1)}}, \Psi_{\underline{\ell}^{(i_4)}} \rangle \bigg\vert
    = \bigg\vert \sum_{i_1, i_4 = 1}^n
        \langle e_{\ell_{i_1}}, \xi^2 e_{\ell_{i_4}} \rangle
        \langle \Psi_{\underline{\ell}^{(i_1)}}, \Psi_{\underline{\ell}^{(i_4)}} \rangle \bigg\vert \\
    &\le \sum_{i_1, i_4 = 1}^n
        \Vert \xi e_{\ell_{i_1}} \Vert
        \Vert \xi e_{\ell_{i_4}} \Vert
        \Vert \Psi_{\underline{\ell}^{(i_1)}} \Vert
        \Vert \Psi_{\underline{\ell}^{(i_4)}} \Vert
    < \infty \;,
\end{aligned}
\end{equation}
where we recall that $ e_\ell \in \dom(\xi) $, so $ \Vert \xi e_\ell \Vert < \infty $. Similarly, since $\Vert \omega^{-1} f_j \Vert < \infty $,
\begin{equation}
\begin{aligned}
    |\III|
    &= \bigg\vert \sum_{i_1, i_4 = 1}^n
        \sum_{j, j' = 1}^N
        \Big\langle \xi e_{\ell_1'}, \frac{f_j}{\omega} \Big\rangle
        \Big\langle \frac{f_{j'}}{\omega}, \xi e_{\ell_4'} \Big\rangle
        \langle \Psi_{\underline{\ell}^{(i_1)}}, [B_j, B_{j'}^*] \Psi_{\underline{\ell}^{(i_4)}} \rangle \bigg\vert \\
    &\le 2 \sum_{i_1, i_4 = 1}^n
        \sum_{j, j' = 1}^N
        \Vert \xi e_{\ell_1'} \Vert
        \Vert \omega^{-1} f_j \Vert
        \Vert \omega^{-1} f_{j'} \Vert
        \Vert \xi e_{\ell_4'} \Vert
        \Vert B_j \Vert \Vert B_{j'} \Vert
        \Vert \Psi_{\underline{\ell}^{(i_1)}} \Vert
        \Vert \Psi_{\underline{\ell}^{(i_4)}} \Vert
    < \infty \;.
\end{aligned}
\end{equation}
Finally, to bound $ \I $, for $ i < i' $, we introduce the notation
\begin{equation}
    \Psi_{\underline{\ell}^{(i,i')}}
    := a^*_{\ell_1} \ldots
        a^*_{\ell_{i-1}} a^*_{\ell_{i+1}} \ldots
        a^*_{\ell_{i'-1}} a^*_{\ell_{i'+1}} \ldots
        a^*_{\ell_n} T (v \otimes \Omega) \qquad \textnormal{and} \qquad
    \Psi_{\underline{\ell}^{(i',i)}} := \Psi_{\underline{\ell}^{(i,i')}} \;.
\end{equation}
So $ \Psi_{\underline{\ell}^{(i,i')}} $ results from $ \Psi_{\underline{\ell}} $ by removing $ a^*_{\ell_i} $ and $ a^*_{\ell_{i'}} $. Then,
\begin{equation} \label{eq:bound_I}
\begin{aligned}
    |\I|
    &= \bigg\vert \sum_{i_1 \neq i_3} \sum_{i_2 \neq i_4}
        \langle e_{\ell_{i_1}}, \xi e_{\ell_{i_2}} \rangle 
        \langle e_{\ell_{i_3}}, \xi e_{\ell_{i_4}} \rangle 
        \langle \Psi_{\underline{\ell}^{(i_1,i_3)}}, \Psi_{\underline{\ell}^{(i_2,i_4)}} \rangle \bigg\vert \\
    &\le \sum_{i_1 \neq i_3} \sum_{i_2 \neq i_4}
        \Vert \xi^{\frac 12} e_{\ell_{i_1}} \Vert
        \Vert \xi^{\frac 12} e_{\ell_{i_2}} \Vert
        \Vert \xi^{\frac 12} e_{\ell_{i_3}} \Vert
        \Vert \xi^{\frac 12} e_{\ell_{i_4}} \Vert
        \Vert \Psi_{\underline{\ell}^{(i_1,i_3)}} \Vert
        \Vert \Psi_{\underline{\ell}^{(i_2,i_4)}} \Vert
    < \infty \;.
\end{aligned}
\end{equation}
So $ \di \hat{\Gamma} (\xi) \Psi_{\underline{\ell}} \in \hfrak $ and $ \di \hat{\Gamma} (\xi) $ is indeed well-defined on $ \mathcal{D} $.
\end{proof}

\begin{proof}[Proof of Theorem~\ref{thm:HE2def_T}]
Denseness of $ \mathcal{D} $ was established in Proposition~\ref{prop:D_dense} and well-definedness of $ (H-K) = \di\hat{\Gamma}(\omega) $ on $ \mathcal{D} $ follows from Lemma~\ref{lem:di_hat_Gamma}. Since $ \omega \ge 0 $ and $ K $ is bounded, $ H $ is also semi-bounded from below. Therefore, by Friedrichs' theorem~\cite[Thm.~2.13]{teschl2014mathematical}, $ H $ admits a self-adjoint extension.
\end{proof}

\section{Proof of norm resolvent convergence}
\label{sec:proof_normresolvent}

\noindent We begin by recalling the expression~\eqref{eq:HLambda} of the cut-off Hamiltonian:
\begin{equation*}
    H_\Lambda - E_\Lambda = K + \di \Gamma (\omega)
        + \sum_{j=1}^N (B_j a^*(f_{j,\Lambda}) + B_j^* a(f_{j,\Lambda}) )
        + \sum_{j,j'=1}^N B_j^* B_{j'} \int \frac{\overline{f_{j,\Lambda}(k)} f_{j',\Lambda}(k)}{\omega(k)} \di k \;,
\end{equation*}
where the cut-off form factors $(f_{j,\Lambda})_{j=1,\dots,N}\subset L^2(X)$ are chosen in such a way that $ \frac{f_{j,\Lambda}}{\omega} \to \frac{f_j}{\omega} $ in $ L^2(X) $. Correspondingly, the cut-off resolvent and full resolvent are given by
\begin{equation} \label{eq:R_Lambda_z}
    R_\Lambda(z) := (H_\Lambda - E_\Lambda - z)^{-1} \;, \qquad
    R(z) := (H - z)^{-1} \;,
\end{equation}
respectively. In analogy to $ \hat a(g) $ with $ g \in L^2(X) $ and to $ \hat{a}_\ell $, defined in~\eqref{eq:ahat_ell_T}, we introduce the operators $ \hat{a}_\Lambda(g) $ and $ \hat{a}_{\Lambda,\ell} $ via
\begin{equation}
     \hat{a}_\Lambda(g) := a(g) + \sum_{j=1}^N B_j \Big\langle \frac{g}{\omega}, f_{j,\Lambda} \Big\rangle \;, \qquad
     \textnormal{for } \frac{g}{\omega} \in L^2(X) \;, \qquad
     \hat{a}_{\Lambda,\ell} := \hat{a}_\Lambda(e_\ell) \;,
\end{equation}
with adjoints $ \hat{a}^*_\Lambda(g) $ and $ \hat{a}_{\Lambda,\ell}^* $. This allows us to \textit{formally} write the resolvent difference as
\begin{equation} \label{eq:resolventdifference}
    R_\Lambda(z) - R(z)
    = \sum_{j=1}^N R_\Lambda(z) \Big( B_j^* \hat{a}(f_j - f_{j,\Lambda}) + \hat{a}_\Lambda^*(f_j - f_{j,\Lambda}) B_j \Big) R(z) \;.
\end{equation}
To make mathematical sense of the formal expression above, first notice that $ \hat{a}_\Lambda(g) $ is indeed well-defined on a dense subspace of $ \hfrak $ for $ \frac{g}{\omega} \in L^2 $, since 
\begin{equation}
    \Big\langle \frac{g}{\omega}, f_{j,\Lambda} \Big\rangle \le \Big\Vert \frac{g}{\omega} \Big\Vert \Vert f_{j,\Lambda} \Vert \;.
\end{equation}

One can then define its adjoint $ \hat{a}^*_\Lambda(g) $ using suitable Hilbert space riggings as in~\cite[Sect.~3]{lonigro2022generalized}, thus ensuring well-definedness of $\hat{a}_\Lambda^*(f_j - f_{j,\Lambda})$ even though $ (f_j - f_{j,\Lambda}) \notin L^2(X) $. We will explicitly construct such a suitable rigging in~\eqref{eq:rigging}. However, $ \hat{a}(f_j - f_{j,\Lambda}) $ contains a possibly divergent integral $ \langle f_j, \frac{f_j}{\omega} \rangle$. We therefore need to define $\hat{a}(g)$ for $g\not\in L^2(X)$, but $\omega^{-1}g\in L^2(X)$. To that end we define
\begin{equation}\label{eq:hat_a_critical}
    \hat{a}(g):=\sum_{\ell\in \mathbb{N}}\langle g,e_\ell\rangle\hat{a}_\ell\,.
\end{equation}
A priori, it is not evident that this operator is well defined on a dense subset of $\hfrak$, as the series might fail to converge. However, following an analogous argument as in the proof of Lemma~\ref{lem:di_hat_Gamma}, it can be easily seen that the sum becomes finite when applied to elements in $\mathcal{D} $.
Furthermore, if $g\in L^2(X)$, this definition coincides with the one given in~\eqref{eq:ahat_ell_T}, since in that case $g$ admits an expansion in the basis $(e_\ell)_{\ell\in\mathbb{N}}$.  

To define the rigging for $ \hat{a}_\Lambda^*(g)$, let us introduce the following shifted second quantization operators, which are a cut-off version of $ \di \hat\Gamma(\xi) $ in~\eqref{eq:di_hat_Gamma}:
\begin{equation}
    \di \hat{\Gamma}_\Lambda (\xi)
    := \sum_{\ell, \ell' \in \mathbb{N}} \langle e_\ell, \xi e_{\ell'} \rangle \hat{a}_\Lambda(e_\ell)^* \hat{a}_\Lambda(e_{\ell'}) \;.
\end{equation}
Similarly as in Lemma~\ref{lem:di_hat_Gamma}, one can show that, whenever $ \omega^{-1} f_{j,\Lambda} \in L^2(X) $, then $ \di \hat{\Gamma}_\Lambda (\xi) $ is well-defined on
\begin{equation} \label{eq:cD_Lambda}
\begin{aligned}
    \mathcal{D}_\Lambda &:= \mathrm{Span}\{ a^*_\Lambda(e_{\ell_1}) \ldots a^*_\Lambda(e_{\ell_n}) T_\Lambda (v \otimes \Omega) \in \hfrak : n \in \mathbb{N}_0, \; \underline{\ell} \in \mathbb{N}^n , \; v \in \mathbb{C}^D \} \;, \\
    T_\Lambda &:= \e^{B_\Lambda} \;, \qquad
    B_\Lambda := -\sum_{j=1}^N B_j  a^* \left( \frac{f_{j,\Lambda}}{\omega} \right) \;, \qquad
    \underline{\ell} = (\ell_1, \ldots, \ell_n) \;.
\end{aligned}
\end{equation}
Our strategy is now to prove the following bounds: 
 $$\Vert \hat{a}_\sharp(g) \Psi \Vert \le \Vert \omega^{-1} g \Vert \Vert \di \hat\Gamma_\sharp(\omega^2)^{\frac 12} \Psi \Vert \le C \Vert \omega^{-1} g \Vert \Vert \di \hat\Gamma_\sharp(\omega) \Psi \Vert\,, $$
with $\sharp\in\{\Lambda,\,\cdot\,\}$. The first inequality holds in general whenever $g\omega^{-1}\in L^2(X)$, whereas the second bound will require a sufficiently small coupling. To address this, we will prove these bounds for the restriction of $g$ to some measurable set $ S\subseteq X $:

\begin{lemma} \label{lem:ahatconversion}
Let $ \omega $, $ f_j $, and $ B_j $ satisfy Assumption~\ref{as:strong}, and let $ f_{j,\Lambda} \in L^2(X) $. Let $ g: X \to \mathbb{C} $ be a measurable function satisfying $ \frac{g}{\omega} \in L^2(X) $. Then, for any measurable set $ S\subseteq X $ with indicator function $ \chi_S $ and any $ \Psi \in \hfrak $, the following inequalities hold:
\begin{equation} \label{eq:ahatconversion}
    \Vert \hat{a}_\Lambda(g \chi_S) \Psi \Vert
    \le \Big\Vert \frac{g}{\omega} \chi_S \Big\Vert
        \Vert \di \hat{\Gamma}_\Lambda (\omega^2 \chi_S)^{\frac 12} \Psi \Vert \quad  \text{and}\quad \Vert \hat{a}(g \chi_S) \Psi \Vert
    \le \Big\Vert \frac{g}{\omega} \chi_S \Big\Vert
        \Vert \di \hat{\Gamma} (\omega^2 \chi_S)^{\frac 12} \Psi \Vert \;.
\end{equation}
\end{lemma}

\begin{proof}
By Definition~\eqref{eq:hat_a_critical} and the Cauchy--Schwarz inequality, we have
\begin{equation} \label{eq:ahat_CS_bound}
\begin{aligned}
    &\Vert \hat{a}(g \chi_S) \Psi \Vert
    = \Big\Vert \sum_{\ell, \ell' \in \mathbb{N}} \Big\langle \frac{g}{\omega} \chi_S, e_{\ell'} \Big\rangle \langle e_{\ell'}, \omega \chi_S e_\ell \rangle \hat{a}_\ell \Psi \Big\Vert \\
    &\le \Big( \sum_{\ell' \in \mathbb{N}} \Big\vert \Big\langle \frac{g}{\omega} \chi_S, e_{\ell'} \Big\rangle \Big\vert^2 \Big)^{\frac12}
        \Big( \sum_{\ell' \in \mathbb{N}} \Big\Vert \sum_{\ell \in \mathbb{N}} \langle e_{\ell'}, \omega \chi_S e_\ell \rangle \hat{a}_\ell \Psi \Big\Vert^2 \Big)^{\frac 12}
    = \Big\Vert \frac{g}{\omega} \chi_S \Big\Vert
        \Vert \di \hat{\Gamma} (\omega^2 \chi_S)^{\frac 12} \Psi \Vert \;.
\end{aligned}
\end{equation}
The proof for $ \hat{a}_\Lambda(g \chi_S) $ is identical.
\end{proof}

\begin{lemma} \label{lem:HLambdasquare_bound}
    Let $ \omega $, $ f_j $, and $ B_j $ satisfy Assumption~\ref{as:strong} and let $ \frac{f_{j,\Lambda}}{\omega} \to \frac{f_j}{\omega} $ in $ L^2(X) $.
    Then, given $ \varepsilon > 0 $ and $S\subseteq X$ a measurable set such that 
    \begin{equation}
        \left\|B_j\right\|\left\|\frac{f_j}{\omega}\chi_S\right\|<\delta:=\sqrt{\frac{\varepsilon}{8N^2}}\quad\text{ for all }1\leq j\leq N,
    \end{equation}
    there exists some $ \Lambda_{\min} \in \mathbb{N} $ such that, for all $ \Lambda \ge \Lambda_{\min} $, 
\begin{equation}
    \di \hat{\Gamma}(\omega \chi_S)^2 \ge (1-\varepsilon) \; \di \hat{\Gamma}(\omega^2 \chi_S) \;, \qquad
    \di \hat{\Gamma}_\Lambda(\omega \chi_S)^2 \ge (1-\varepsilon) \; \di \hat{\Gamma}_\Lambda(\omega^2 \chi_S) \;.
\end{equation}
Furthermore, $ \di \Gamma(\omega \chi_S)^2 \ge \di \Gamma(\omega^2 \chi_S) $.
\end{lemma}

\begin{proof}
First notice that, since $ \omega^{-1}f_{j,\Lambda}\to \omega^{-1}f_j $, we have $ \Vert B_j \Vert \Vert \omega^{-1}f_{j,\Lambda} \chi_S\Vert < 2 \delta $ for $ \Lambda $ large enough. We first prove the lower bound on $ \di \hat{\Gamma}_\Lambda(\omega \chi_S)^2 $ for such $ \Lambda $. Using the CCR and $ [B_j,B_{j'}] = 0 $, one easily verifies $ [\hat{a}_{\Lambda,\ell}, \hat{a}_{\Lambda',\ell'}] = 0 $, so
\begin{equation}
\begin{aligned}
    &\di \hat{\Gamma}_\Lambda(\omega \chi_S)^2 \\
    &= \sum_{\ell_1, \ell_2, \ell_3, \ell_4 \in \mathbb{N}}
        \langle e_{\ell_1}, \omega \chi_S e_{\ell_2} \rangle 
        \langle e_{\ell_3}, \omega \chi_S e_{\ell_4} \rangle 
        \big( \hat{a}^*_{\Lambda,\ell_1} \hat{a}^*_{\Lambda,\ell_3} \hat{a}_{\Lambda,\ell_2} \hat{a}_{\Lambda,\ell_4}
        + \hat{a}^*_{\Lambda,\ell_1} [\hat{a}_{\Lambda,\ell_2}, \hat{a}^*_{\Lambda,\ell_3}] \hat{a}_{\Lambda,\ell_4} \big) \;.
\end{aligned}
\end{equation}
The first contribution is non-negative, and evaluating the commutator we get
\begin{equation} \label{eq:HLambda_intermediate_bound}
    \di \hat{\Gamma}_\Lambda(\omega \chi_S)^2
    \ge \di \hat{\Gamma}_\Lambda(\omega^2 \chi_S)
        + \sum_{j,j'=1}^N [B_j,B_{j'}^*] \hat{a}_\Lambda^*(f_{j,\Lambda} \chi_S) \hat{a}_\Lambda(f_{j',\Lambda} \chi_S) \;.
\end{equation}
We bound the second contribution using the Cauchy--Schwarz inequality and Lemma~\ref{lem:ahatconversion}:
\begin{equation}
\begin{aligned}
    &\bigg\langle \Psi, \sum_{j,j'=1}^N [B_j,B_{j'}^*] \hat{a}_\Lambda^*(f_{j,\Lambda} \chi_S) \hat{a}_\Lambda(f_{j',\Lambda} \chi_S) \Psi \bigg\rangle 
    \ge - \sum_{j, j'=1}^N \Vert [B_j,B_{j'}^*] \Vert
        \Vert \hat{a}_\Lambda(f_{j,\Lambda} \chi_S) \Psi \Vert 
        \Vert \hat{a}_\Lambda(f_{j',\Lambda} \chi_S) \Psi \Vert \\
    &\ge - 2 N^2 \Big( \max_{1 \le j \le N} \Vert B_j \Vert \Big\Vert \frac{f_{j,\Lambda}}{\omega} \chi_S \Big\Vert \Big)^2
    \Vert \di \hat{\Gamma}_\Lambda(\omega^2 \chi_S)^{\frac 12} \Psi \Vert^2
    \\
    &
    \ge -8N^2\delta^2\Vert \di \hat{\Gamma}_\Lambda(\omega^2 \chi_S)^{\frac 12} \Psi \Vert^2
    \\
    &
    =-\varepsilon\left\langle \Psi,  \di \hat{\Gamma}_\Lambda(\omega^2 \chi_S)\Psi\right\rangle
    \;,
\end{aligned}
\end{equation}
which gives 
\begin{equation}
    \sum_{j,j'=1}^N [B_j,B_{j'}^*] \hat{a}_\Lambda^*(f_{j,\Lambda} \chi_S) \hat{a}_\Lambda(f_{j',\Lambda} \chi_S)
    \ge -\varepsilon\, \di \hat{\Gamma}_\Lambda(\omega^2 \chi_S) \;,
\end{equation}
and together with~\eqref{eq:HLambda_intermediate_bound} readily implies the claim. The bound $ \di \hat\Gamma(\omega \chi_S)^2 \ge (1 - \varepsilon) \di \hat\Gamma(\omega^2 \chi_S) $ follows analogously, replacing $ f_{j,\Lambda} $ with $ f_j $. In particular, $ \hat{a}(f_j \chi_S) $ is well-defined by Lemma~\ref{lem:ahatconversion}, even if $ \frac{f_j}{\omega} \chi_S \in L^2(X) $ but $ f_j \chi_S \notin L^2(X) $. Also, the bound $ \di \Gamma(\omega \chi_S)^2 \ge \di \Gamma(\omega^2 \chi_S) $ is analogous, where the $ [B_j, B_{j'}] $-term in~\eqref{eq:HLambda_intermediate_bound} is absent.
\end{proof}

The following lemma will allow us to remove the restriction to small coupling, and thus finally to define $ \hat{a}(g) $ on the desired domain $ \dom(H) = \dom(\di \hat\Gamma(\omega)) $ for $g\notin L^2(X)$ with $\omega^{-1}g\in L^2(X)$.

\begin{lemma} \label{lem:ahat_R_bound}
Let $ \omega $, $ f_j $, and $ B_j $ satisfy Assumption~\ref{as:strong} and let $ \frac{f_{j,\Lambda}}{\omega} \to \frac{f_j}{\omega} $ in $ L^2(X) $. Then, for $ \Re(-z) > 0 $ and $ \Lambda_{\min} \in \mathbb{N} $ large enough, for all $ \frac{g}{\omega} \in L^2(X) $ and $ \Lambda \ge \Lambda_{\min} $, there exists a constant $ C > 0 $ such that
\begin{equation}
    \Vert \hat{a}(g) R(z) \Vert
    \le C \Big\Vert \frac{g}{\omega} \Big\Vert \;, \qquad
    \Vert \hat{a}_\Lambda(g) R_\Lambda(z) \Vert
    \le C \Big\Vert \frac{g}{\omega} \Big\Vert \;.
\end{equation}
\end{lemma}
In particular, $ \hat{a}(g) R(z) $ and $ \hat{a}_\Lambda(g) R_\Lambda(z) $ are bounded operators defined on all of $ \hfrak $.
\begin{proof}
We want to employ Lemma~\ref{lem:HLambdasquare_bound}, which only holds for small coupling. To do so, we split $ g $ as follows:
\begin{enumerate}
    \item Fix an energy cut-off $ \kappa > 0 $, and set with $ S := \{k \in X : \omega(k) > \kappa \} $, such that for some given $ \delta > 0 $ to be fixed later, we have $ \Vert \omega^{-1}f_j\chi_S \Vert < \delta $ for $ 1 \le j \le N $;
    \item Then, define the UV tail $ g^> := g \chi_S $ and the finite-energy contribution $ g^{\le} := g (1-\chi_S) $, so that
\begin{equation} \label{eq:ag_split}
    \Vert \hat{a}_\Lambda(g) R_\Lambda(z) \Vert
    \le \Vert \hat{a}_\Lambda(g^>) R_\Lambda(z) \Vert
        + \Vert \hat{a}_\Lambda(g^{\le}) R_\Lambda(z) \Vert \;.
\end{equation}
\end{enumerate}
We will bound the two terms on the right-hand side of Eq.~\eqref{eq:ag_split} separately, starting from the first one. We introduce $ \omega^> := \omega \chi_S $, $ \omega^{\le} := \omega(1-\chi_S) $ and $ f_{j,\Lambda}^> := f_{j,\Lambda} \chi_S $, $ f_{j,\Lambda}^{\le} := f_{j,\Lambda}(1-\chi_S) $, and split $ H_\Lambda - E_\Lambda $ as
\begin{equation} \label{eq:HLambda_split}
\begin{aligned}
    H_\Lambda - E_\Lambda
    &= K + \di \hat{\Gamma}_\Lambda(\omega^>) + \di \hat{\Gamma}_\Lambda(\omega^{\le}) \\
    &= K + \di \hat{\Gamma}_\Lambda(\omega^>) + \di \Gamma(\omega^{\le}) + H_{\mathrm{int},\Lambda}^{\le} \;, \\
    H_{\mathrm{int},\Lambda}^{\le}
    &:= \sum_{j=1}^N B_j^* a(f_{j,\Lambda}^{\le})
        + \sum_{j=1}^N B_j a^*(f_{j,\Lambda}^{\le})
        + \sum_{j,j'=1}^N B_j^* B_{j'} \int \frac{\overline{f_{j,\Lambda}^{\le}(k)} f_{j',\Lambda}^{\le}(k)}{\omega(k)} \di k \;.
\end{aligned}
\end{equation}
Then, using Lemma~\ref{lem:ahatconversion} and~\ref{lem:HLambdasquare_bound} with $ \varepsilon = \frac 12 $, for $ \delta $ small enough we get
\begin{equation} \label{eq:aR_UV_reduction}
    \Vert \hat{a}_\Lambda(g^>) R_\Lambda(z)\Psi \Vert
    \le \Big\Vert \frac{g}{\omega} \chi_S \Big\Vert \Vert \di \hat{\Gamma}_\Lambda(\omega^2 \chi_S)^{\frac 12} R_\Lambda(z) \Psi \Vert^2
    \le 2 \Big\Vert \frac{g}{\omega} \Big\Vert
    \Vert \di \hat{\Gamma}_\Lambda(\omega^>) R_\Lambda(z) \Psi \Vert^2 \;,
\end{equation}
uniformly in $ \Lambda \ge \Lambda_{\min} $ for some $ \Lambda_{\min} \in \mathbb{N} $ not depending on $ g $. Introducing the reduced resolvent
\begin{equation} \label{eq:R_>}
    R_\Lambda^>(z) := \big( \di \Gamma(\omega^{\le}) + \di \hat{\Gamma}_\Lambda(\omega^>) - z \big)^{-1} \;,
\end{equation}
we then conclude, by means of a Neumann expansion,
\begin{equation} \label{eq:R_to_R_>}
    \Vert \di \hat{\Gamma}_\Lambda(\omega^>) R_\Lambda(z) \Vert
    \le \Vert \di \hat{\Gamma}_\Lambda(\omega^>) R_\Lambda^>(z) \Vert
        \Big( 1 + \sum_{n=1}^\infty \Vert (H_{\mathrm{int},\Lambda}^{\le} + K) R_\Lambda^>(z) \Vert^n \Big) \;.
\end{equation}
For the first factor, as $ \omega^> $ and $ \omega^{\le} $ have disjoint supports, one easily checks $ [\di \Gamma(\omega^{\le}), \di \hat{\Gamma}_\Lambda(\omega^>) ] = 0 $. So there exists a common spectral measure for both operators, and with $ \di \Gamma(\omega^{\le}) \ge 0 $ we conclude $ \Vert \di \hat{\Gamma}_\Lambda(\omega^>) R_\Lambda^>(z) \Vert < 1 $ for $ \Re(-z) > 0 $ and uniformly in $ \Lambda $. For the second factor, we use the inequality $ \Vert f_{j,\Lambda}^{\le} \Vert \le \kappa \Vert \omega^{-1}f_{j,\Lambda} \Vert $ and the fact that the latter is uniformly bounded in $ \Lambda $. Therefore, $ H_{\mathrm{int},\Lambda}^{\le} + K $ is infinitesimally Kato-bounded against $ \di \Gamma(\omega^{\le}) $, see for instance~\cite[Corollary~5.10]{arai2018analysis}, in the sense that, for any $ \varepsilon > 0 $, there exists some $ b_\varepsilon > 0 $ such that uniformly in $ \Lambda $,
\begin{equation}
    \Vert (H_{\mathrm{int},\Lambda}^{\le} + K) R_\Lambda^>(z) \Vert
    \le \varepsilon \Vert \di \Gamma(\omega^{\le}) R_\Lambda^>(z) \Vert + b_\varepsilon \Vert R_\Lambda^>(z) \Vert \;.
\end{equation}
By the same spectral calculus argument as above, with $ \di \hat{\Gamma}_\Lambda(\omega^>) \ge 0 $, we get $ \Vert \di \Gamma(\omega^{\le}) R_\Lambda^>(z) \Vert < 1 $ for $ \Re(-z) > 0 $. Furthermore, choosing $ \Re(-z) $ large enough renders $ b_\varepsilon \Vert R_\Lambda^>(z) \Vert < \varepsilon $ uniformly in $ \Lambda $. So with $ \varepsilon $ small enough and $ \Lambda $ large enough, we have for some $ C > 0 $\footnote{In this section, $ C > 0 $ denotes a constant independent of $ \Lambda $ and $ g $, which may vary from line to line.}
\begin{equation} \label{eq:ag_>_bound}
    \Vert \di \hat{\Gamma}_\Lambda(\omega^>) R_\Lambda(z) \Vert
    \le 1 + \sum_{n=1}^\infty (2 \varepsilon)^n
    \le 2 \qquad \overset{\eqref{eq:aR_UV_reduction}}{\Rightarrow} \qquad
    \Vert \hat{a}_\Lambda(g^>) R_\Lambda(z) \Vert
    \le C \Big\Vert \frac{g}{\omega} \Big\Vert \;.
\end{equation}

Now we bound the second term on the right-hand side of~\eqref{eq:ag_split}. Recall that $ \hat{a}_\Lambda(g^{\le}) = a(g^{\le}) + \sum_{j=1}^N B_j \langle g^{\le}, \frac{f_{j,\Lambda}}{\omega} \rangle $. With $ \Vert a(g^{\le}) \Psi \Vert \le \Big\Vert \frac{g^{\le}}{\omega} \Big\Vert \Vert \di \Gamma(\omega^2)^{\frac 12} \Psi \Vert $ and $ \di \Gamma(\omega^2) \le \di \Gamma (\omega)^2 $ (see Lemma~\ref{lem:HLambdasquare_bound}), we conclude
\begin{equation}
    \Vert \hat{a}_\Lambda(g^{\le}) R_\Lambda(z) \Vert
    \le \bigg\Vert \frac{g}{\omega} \bigg\Vert
        \Vert \di \Gamma(\omega^{\le}) R_\Lambda(z) \Vert
        + \sum_{j=1}^N \kappa \Vert B_j \Vert
            \bigg\Vert \frac{g}{\omega} \bigg\Vert
            \bigg\Vert \frac{f_{j,\Lambda}^{\le}}{\omega} \bigg\Vert
            \Vert R_\Lambda(z) \Vert \;.
\end{equation}
As in~\eqref{eq:R_to_R_>}--\eqref{eq:ag_>_bound}, we may bound $ \Vert \di \Gamma(\omega^{\le}) R_\Lambda(z) \Vert < C $. Furthermore, $ \Big\Vert \frac{f_{j,\Lambda}^{\le}}{\omega} \Big\Vert \le C $, so
\begin{equation}\label{eq:ag_<_bound}
    \Vert \hat{a}_\Lambda(g^{\le}) R_\Lambda(z) \Vert
    \le C \Big\Vert \frac{g}{\omega} \Big\Vert \;.
\end{equation}

We finally managed to bound both terms on the right-hand side of Eq.~\eqref{eq:ag_split}, see Eqs.~\eqref{eq:ag_>_bound} and~\eqref{eq:ag_<_bound}. Plugging them into~\eqref{eq:ag_split} renders the desired bound on $ \Vert \hat{a}_\Lambda(g) R_\Lambda(z) \Vert $. The bound on $ \Vert \hat{a}(g) R(z) \Vert $ is analogous under replacement $ f_{j,\Lambda} \to f_j $, where in particular $ \Vert f_j^{\le} \Vert \le \kappa \Vert \omega^{-1} f_j \Vert $, so $ f_j^{\le} \in L^2 $.
\end{proof}

We can now define the Hilbert space riggings $ \hfrak_{\Lambda,+} \subset \hfrak \subset \hfrak_{\Lambda,-} $ and $ \hfrak_+ \subset \hfrak \subset \hfrak_- $ via
\begin{equation} \label{eq:rigging}
    \hfrak_{\Lambda,+}
    := \dom(\di \hat\Gamma_\Lambda(\omega) )
    = \dom(H_\Lambda) \;, \qquad
    \hfrak_+
    := \dom(\di \hat\Gamma(\omega) )
    = \dom(H) \;,
\end{equation}
and $ \hfrak_{\Lambda,-}, \hfrak_- $ being their dual spaces, cf.~\cite[Sect. 3]{lonigro2022generalized}.
Since $ R_\Lambda(z): \hfrak \to \hfrak_{\Lambda,+} $ and $ R(z): \hfrak \to \hfrak_+ $ are injective, Lemma~\ref{lem:ahat_R_bound} readily allows us to define the continuous maps
\begin{equation} \label{eq:rigging_a}
    \hat{a}_\Lambda(g): \hfrak_{\Lambda,+} \to \hfrak \;, \qquad
    \hat{a}(g): \hfrak_+ \to \hfrak \;, \qquad \textnormal{for } \frac{g}{\omega} \in L^2(X) \;,
\end{equation}
as well as their adjoints 
\begin{equation} \label{eq:rigging_astar}
    \hat{a}^*_\Lambda(g): \hfrak \to \hfrak_{\Lambda,-} \;, \qquad
    \hat{a}^*(g): \hfrak \to \hfrak_- \;, \qquad \textnormal{for } \frac{g}{\omega} \in L^2(X) \;.
\end{equation}
We can now proceed with the proof of norm resolvent convergence.

\begin{proof}[Proof of Theorem~\ref{thm:normresolvent}]
One easily verifies the following strong operator identity\footnote{
If $ \dom(H_\Lambda) \cap \dom(H) $ is not dense in $ \hfrak $, we proceed as follows: Recall the rigging $ \hfrak_{\Lambda,+} \subset \hfrak \subset \hfrak_{\Lambda_-} $ from~\eqref{eq:rigging}. We then extend $ H_\Lambda $ and $ R_\Lambda(z) $ to $ H_\Lambda: \hfrak \to \hfrak_{\Lambda,-} $ and $ R_\Lambda(z): \hfrak_{\Lambda,-} \to \hfrak $. So for $ \Psi \in \dom(H) $, we have $ (H - H_\Lambda) \Psi \in \hfrak_{\Lambda,-} $ and $ R_\Lambda(z) (H - H_\Lambda) \Psi \in \hfrak $.} on $ \hfrak $:
\begin{equation}
    (R_\Lambda(z) - R(z))
    = R_\Lambda(z) (H - H_\Lambda - E_\Lambda) R(z)
    = \sum_{j=1}^N R_\Lambda(z) \Big( B_j^* \hat{a}(f_j - f_{j,\Lambda}) + \hat{a}_\Lambda^*(f_j - f_{j,\Lambda}) B_j \Big) R(z) \;.
\end{equation}
By using Lemma~\ref{lem:ahat_R_bound}, we can then bound the left-hand side as follows:
\begin{equation}
\begin{aligned}
    &\Vert R_\Lambda(z) - R(z) \Vert
    \le \sum_{j=1}^N \Big\Vert R_\Lambda(z) \Big( B_j^* \hat{a}(f_j - f_{j,\Lambda}) + \hat{a}_\Lambda^*(f_j - f_{j,\Lambda}) B_j \Big) R(z) \Big\Vert \\
    &\le C \sum_{j=1}^N \Vert B_j \Vert
        \Big( \big\Vert \hat{a}(f_j - f_{j,\Lambda}) R(z) \big\Vert
        + \big\Vert R_\Lambda(z) \hat{a}_\Lambda^*(f_j - f_{j,\Lambda}) \big\Vert \Big) \\
    &\le C \sum_{j=1}^N \Vert B_j \Vert \bigg\Vert \frac{f_j}{\omega} - \frac{f_{j,\Lambda}}{\omega} \bigg\Vert \to 0 \qquad \textnormal{as } \Lambda \to \infty \;,
\end{aligned}
\end{equation}
where we used $ \Vert R_\Lambda(z) \hat{a}_\Lambda^*(f_j - f_{j,\Lambda}) \Vert = \Vert \hat{a}_\Lambda(f_j - f_{j,\Lambda}) R_\Lambda(z) \Vert $, as both operators are bounded and adjoints of each other. This concludes the proof.
\end{proof}

\begin{remark}
At this point, the reader may wonder why we did not try to prove resolvent convergence by simply expanding 
\begin{equation*}
    (H_\Lambda - z)^{-1} - (H_{\Lambda'} - z)^{-1}
    = \sum_{j=1}^N (H_\Lambda - z)^{-1} \big( B_j a^*(f_{j,\Lambda'} - f_{j,\Lambda}) + B_j^* a(f_{j,\Lambda'} - f_{j,\Lambda}) \big) (H_{\Lambda'} - z)^{-1} \;,
\end{equation*}
and then bound $ \Vert a(f_{j,\Lambda'} - f_{j,\Lambda}) (H_{\Lambda'} - z)^{-1} \Vert $ uniformly in $ \Lambda $ and $ \Lambda' $. In fact, this expansion would not allow incorporating the self-energy $ E_\Lambda $, so we would not expect it to give a reasonable renormalized Hamiltonian in case $ E_\Lambda \to \infty $.
And indeed, while we bounded $ \hat{a}_\Lambda(g) $ against $ \di \hat\Gamma_\Lambda(\omega) = H_\Lambda - E_\Lambda $, the analogous bound for $ a(g) $ would only be against $ \di \Gamma(\omega) $, which differs from $ H_\Lambda $ by a large perturbation.
\end{remark}

\appendix
\section{Proof of triviality for UV-divergences of Case 3}
\label{sec:triviality}

As in~\cite{dam2020asymptotics}, we will decompose the Hamiltonian into fibers and conjugate each fiber with a generalized Weyl transformation $ W(s,U): \fock \to \fock $ where $ s \in  L^2(X) $, $ U:\fock \to \fock $ is unitary, and
\begin{equation} \label{eq:WsU}
    W(s,U) := \e^{a^*(s) - a(s)} \Gamma(U) \;, \qquad
    (\Gamma(U) \psi)^{(n)} := U^{\otimes n} \psi \;.
\end{equation}
From~\cite[(2.2)]{dam2020asymptotics}, we retrieve the following generalized Weyl relations:
\begin{equation} \label{eq:Weylrelations}
    W(s_1,U_1) W(s_2,U_2) = \e^{-i \mathrm{Im}\langle s_1, U_1 s_2 \rangle} W(s_1 + U_1 s_2, U_1 U_2) \;.
\end{equation}
These allowed the authors of~\cite{dam2020asymptotics} to calculate the conjugated fibers for the standard spin--boson model, where they take the form $ \di\Gamma(\omega) + \eta W(g,-1) $ for some $ \eta \in \mathbb{R} $ and $ g \in  L^2(X) $. The authors then proved \textit{triviality} on the renormalized model, in the following sense:
\begin{equation} \label{eq:trivialitystandard}
    \di\Gamma(\omega) + \eta W(g,-1) \to \di\Gamma(\omega) \quad
    \textnormal{ in the norm resolvent sense as } \quad
    \Vert g \Vert \to \infty \;.
\end{equation}
It is possible to generalize this result to any phase $ \e^{i\alpha} \neq 1 $ instead of $ -1 $, and a finite sum of perturbations. To this purpose, we will employ the following lemma:

\begin{lemma}[Linear combinations of perturbations] \label{lem:linearcomb_pert}
Let $ H $ be a self-adjoint operator on $ \hfrak $. For $ M \in \mathbb{N} $, $ 1 \le m \le M $, let $ (V_{m,\Lambda})_{\Lambda \in \mathbb{N}} \subset \mathcal{B}(\hfrak) $ be a family of operators, uniformly bounded in $ \Lambda $, such that $ (H+V_{m,\Lambda}) \to H $ in norm resolvent sense as $ \Lambda \to \infty $. Let $ (\lambda_m)_{m=1}^M \in \mathbb{C}^M $. Then, we also have
\begin{equation} \label{eq:linearcomb_pert}
    H + \sum_{m=1}^M \lambda_m V_{m,\Lambda} \to H \;,
\end{equation}
in the norm resolvent sense as $ \Lambda \to \infty $.
\end{lemma}
\begin{proof}
First, we note that, for any $ z \in \mathbb{C} $ with $ |\mathrm{Im}(z)| $ large enough, the first resolvent identity yields
\begin{equation}
\begin{aligned}
    \Vert (H + \lambda_m V_{m,\Lambda} - z)^{-1} - (H - z)^{-1} \Vert
    &= \Vert (H + \lambda_m V_{m,\Lambda} - z)^{-1} \lambda_m V_{m,\Lambda} (H - z)^{-1} \Vert \\
    &= |\lambda_m| \Vert (H + V_{m,\Lambda} - z)^{-1} - (H - z)^{-1} \Vert \;,
\end{aligned}
\end{equation}
so also $ (H + \lambda_m V_{m,\Lambda}) \to H $ in the norm resolvent sense. To conclude the proof it suffices to show $ (H + V_{1,\Lambda} + V_{2,\Lambda}) \to H $ in the norm resolvent sense, as an iteration of the additivity argument will then yield~\eqref{eq:linearcomb_pert}. To this end, we expand
\begin{equation} \label{eq:resolventnormdifference_V1V2}
\begin{aligned}
    &\Vert (H + V_{1,\Lambda} + V_{2,\Lambda} - z)^{-1} - (H - z)^{-1} \Vert \\
    &\le \Vert (H + V_{1,\Lambda} + V_{2,\Lambda} - z)^{-1} - (H + V_{1,\Lambda} - z)^{-1} \Vert
        + \underbrace{\Vert (H + V_{1,\Lambda} - z)^{-1} - (H - z)^{-1} \Vert}_{\to 0} \;,
\end{aligned}
\end{equation}
so we must show that the following term vanishes:
\begin{equation} \label{eq:V2V1bound}
\begin{aligned}
    &\Vert (H + V_{1,\Lambda} + V_{2,\Lambda} - z)^{-1} - (H + V_{1,\Lambda} - z)^{-1} \Vert \\
    &= \Vert (H + V_{1,\Lambda} + V_{2,\Lambda} - z)^{-1} V_{2,\Lambda} (H + V_{1,\Lambda} - z)^{-1} \Vert \\
    &\le  \Vert (H + V_{1,\Lambda} + V_{2,\Lambda} - z)^{-1} - (H + V_{2,\Lambda} - z)^{-1} \Vert 
        \underbrace{\Vert V_{2,\Lambda} (H + V_{1,\Lambda} - z)^{-1} \Vert}_{< 1/2} \\
        & \quad + \underbrace{\Vert (H + V_{2,\Lambda} - z)^{-1} V_{2,\Lambda} \Vert}_{\le C} 
        \underbrace{\Vert (H + V_{1,\Lambda} - z)^{-1} - (H - z)^{-1} \Vert}_{\to 0} \\
        & \quad + \underbrace{\Vert (H + V_{2,\Lambda} - z)^{-1} V_{2,\Lambda} (H - z)^{-1} \Vert}_{\to 0} \\
    &\le \frac 12 \Vert (H + V_{1,\Lambda} + V_{2,\Lambda} - z)^{-1} - (H + V_{2,\Lambda} - z)^{-1} \Vert + o_\Lambda(1)
\end{aligned}
\end{equation}
for $ |\mathrm{Im}(z)| $ large enough, as $ \Vert V_{2,\Lambda} \Vert $ is uniformly bounded in $ \Lambda $. Conversely, swapping the roles of $ V_{1,\Lambda} $ and $ V_{2,\Lambda} $ yields
\begin{equation} \label{eq:V1V2bound}
\begin{aligned}
    &\Vert (H + V_{1,\Lambda} + V_{2,\Lambda} - z)^{-1} - (H + V_{2,\Lambda} - z)^{-1} \Vert \\
    &\le \frac 12 \Vert (H + V_{1,\Lambda} + V_{2,\Lambda} - z)^{-1} - (H + V_{1,\Lambda} - z)^{-1} \Vert + o_\Lambda(1) \;.
\end{aligned}
\end{equation}
Plugging~\eqref{eq:V1V2bound} into~\eqref{eq:V2V1bound} and bringing both norms to the same side, we then get
\begin{equation} \label{eq:V1V2bound_2}
    \frac 34 \Vert (H + V_{1,\Lambda} + V_{2,\Lambda} - z)^{-1} - (H + V_{1,\Lambda} - z)^{-1} \Vert
    = o_\Lambda(1) \;,
\end{equation}
whence~\eqref{eq:resolventnormdifference_V1V2} vanishes as $ \Lambda \to \infty $.
\end{proof}

The generalization of~\eqref{eq:trivialitystandard} is now the following.

\begin{lemma} \label{lem:trivialitycondition}
For $ M \in \mathbb{N} $ let $ (\eta_m)_{m=1}^M \in \mathbb{C}^M $ and $ (\alpha_m)_{m=1}^M \in (0,2\pi)^M $. For each $ 1 \le m \le M $, let $ (g_{m,\Lambda})_{\Lambda \in \mathbb{R}} \subset  L^2(X) $ such that $ \Vert g_{m,\Lambda} \Vert \to \infty $ as $ \Lambda \to \infty $. Then,
\begin{equation} \label{eq:trivialitycondition}
    \di\Gamma(\omega) + \sum_{m=1}^M \eta_m W(g_{m,\Lambda}, \e^{i\alpha_m}) 
    \to \di \Gamma(\omega)
\end{equation}
in the norm resolvent sense as $ \Lambda \to \infty $.
\end{lemma}
\begin{proof}
For $ M = 1 $ perturbation and $ \alpha = \pi $, the result follows from~\cite[Lemma~5.6]{dam2020asymptotics}, keeping in mind that $ \omega \ge m > 0 $ implies injectivity of $ \omega $ as an operator. One easily checks that the proof also goes through for generic $ \alpha \in (0,2\pi) $.
Since $ \Vert W(g_{m,\Lambda}, \e^{i\alpha_m}) \Vert \le 1 $ uniformly in $ \Lambda $, we can apply Lemma~\ref{lem:linearcomb_pert}, which yields the desired norm resolvent convergence.
\end{proof}

Using this result, we can finally prove Proposition~\ref{prop:triviality}:
   
\begin{proof}[Proof of Proposition~\ref{prop:triviality}]
Assumption~\ref{as:eigenbasis} allows for a convenient fiber decomposition of the Hamiltonian: Define the unitary ``untwisting map'' $ V: \hfrak \to \hfrak $ by setting for any element $ v_k $ of the eigenbasis $ (v_k)_{k=1}^D \subset \mathbb{C}^D $ of $ K $ from Assumption~\ref{as:eigenbasis}, any 
$ \psi\in \fock $, and any $ n \in \mathbb{N}_0 $:
\begin{equation}
    V(v_k \otimes P_n \psi) = (B^n v_k) \otimes P_n \psi \;,
\end{equation}
where $ P_n $ is the projection to the $ n $-boson sector. Recall that we assumed $ B v_k = v_{k'} $ and that $ B $ is unitary, so $ B^* v_{k'} = v_k $. Furthermore, we have $ K B^n v_k = \kappa_{k;n} B^n v_k $ for some $ \kappa_{k;n} \in \mathbb{R} $, which allows for the following fiber decomposition of $ H_\Lambda $ (compare~\eqref{eq:HLambda}):
\begin{equation} \label{eq:Fk}
    V^* H_\Lambda V
    = \bigoplus_{k = 1}^D \left( \sum_{n = 0}^\infty \kappa_{k;n} P_n
        + \di \Gamma(\omega)
        + a^*(f_\Lambda) + a(f_\Lambda) \right)
    =: \bigoplus_{k = 1}^D F_{k, \Lambda} \;.
\end{equation}
Now, since $Bv_k=v_{k'}$, there must be some $ M \in \mathbb{N} $, $ M \le D $, possibly depending on $ k $, such that $ B^M v_k = v_k $. Therefore, the sequence $ (\kappa_{k;n})_{n \in \mathbb{N}_0} $ is $ M $-periodic, that is, $ \kappa_{k;n+M} = \kappa_{k;n} $, and we can define the discrete Fourier transform
\begin{equation} \label{eq:etaFT}
    \eta_{k;m} := \frac{1}{M} \sum_{n=0}^{M-1} \kappa_{k;n} \e^{-i \frac{2 \pi}{M} m n} \qquad \Leftrightarrow \qquad
    \kappa_{k;n} = \sum_{m=0}^{M-1} \eta_{k;m} \e^{ i \frac{2 \pi}{M} m n} \;,
\end{equation}
with $ \eta_{k;m+M} = \eta_{k;m} $. Then,
\begin{equation}
    F_{k, \Lambda}
    = \eta_{k;0}
        + \sum_{m=1}^{M-1} \eta_{k;m} \Gamma(\e^{i \frac{2 \pi}{M} m})
        + \di \Gamma(\omega)
        + a^*(f_\Lambda) + a(f_\Lambda) \;.
\end{equation}
Using~\eqref{eq:Weylrelations}, we now conjugate $ F_{k, \Lambda} $ with the Weyl transformation $ W(s_\Lambda, 1)$, $\; s_\Lambda := \frac{f_\Lambda}{\omega} $, and then subtract the self-energy $ E_\Lambda = - \Vert \omega^{-1/2} f_\Lambda \Vert^2 $ :
\begin{equation} \label{eq:fibertransform}
    W(s_\Lambda, 1) (F_{k, \Lambda} - E_\Lambda) W(s_\Lambda, 1)^*
    = \eta_{k;0} + \sum_{m=1}^{M-1} \eta_{k;m} e^{\Im \langle s_\Lambda, e^{i \frac{2 \pi}{M} m} s_\Lambda \rangle} W(s_\Lambda (1 - \e^{i \frac{2 \pi}{M} m}), \e^{i \frac{2 \pi}{M} m}) + \di \Gamma(\omega) \;.
\end{equation}
Applying Lemma~\ref{lem:trivialitycondition} with $ \alpha_m = \frac{2 \pi}{M} m $, $ \eta_m = \eta_{k;m} e^{\Im \langle s_\Lambda, e^{i \frac{2 \pi}{M} m} s_\Lambda \rangle} $, and $ g_{m,\Lambda} := s_\Lambda (1 - \e^{i \alpha_m}) $, where $ \Vert g_{m,\Lambda} \Vert = |1 - \e^{i \alpha_m}| \Vert \omega^{-1} f_\Lambda \Vert \to \infty $ as $ \Lambda \to \infty $, we get that the r.h.s. of~\eqref{eq:fibertransform} converges to $ \eta_{k;0} + \di \Gamma(\omega) $. Thus, with
\begin{equation}
    W_\Lambda, \eta_0: \hfrak \to \hfrak \;, \qquad
    W_\Lambda := (1 \otimes W(s_\Lambda,1)) V^* \;, \qquad
    \eta_0 := \bigoplus_{k=1}^D \eta_{k;0} \;,
\end{equation}
we finally obtain
\begin{equation}
    W_\Lambda (H_\Lambda - E_\Lambda) W_\Lambda^*
    \to \eta_0 + \di \Gamma(\omega)
\end{equation}
as $ \Lambda \to \infty $ in the norm resolvent sense.
\end{proof}
    
\AtNextBibliography{\small}	\DeclareFieldFormat{pages}{#1}\sloppy 
	\printbibliography

\end{document}